\documentclass[11pt]{article}
\usepackage[a4paper,margin=1in]{geometry}
\usepackage[utf8]{inputenc}
\usepackage[T1]{fontenc}
\usepackage{authblk}
\usepackage{amsmath,amssymb,amsfonts,amsthm}
\usepackage{array}
\usepackage{booktabs}
\usepackage{longtable}
\usepackage{graphicx}
\usepackage{tikz}
\usetikzlibrary{positioning, arrows.meta, calc}
\usepackage[authoryear,round]{natbib}
\usepackage[hidelinks]{hyperref}

\def\bs{\boldsymbol}

\def\doubleR{\mathbb{R}}

\def\ci{\perp\!\!\!\!\perp}

\newcommand{\rmref}{\mathrm{ref}}
\newcommand{\obs}{\mathrm{obs}}

\theoremstyle{plain}

\newtheorem{proposition}{Proposition}

\theoremstyle{definition}
\newtheorem{definition}{Definition}

\title{Prognosis-equivalent mapping of clinical measurements via survival analysis}
\author[1]{Kazuharu Harada\thanks{Corresponding author: \texttt{haradak@tokyo-med.ac.jp}}}
\author[1]{Mitsunori Ogawa}
\affil[1]{Department of Health Data Science, Tokyo Medical University,
6-1-1 Shinjuku, Shinjuku-ku, Tokyo 160-8402, Japan}
\date{}

\begin{document}

\maketitle

\begin{abstract}
A continuous clinical measurement recorded in fixed physical units may have different prognostic meaning across patients when its effect depends on a patient-level modifier.
For example, the same tumor diameter may imply markedly different prognosis in an infant and an adult, because patient size can modify its prognostic effect.
We formalize this problem through prognosis-equivalent mapping for time-to-event outcomes.
The resulting estimand maps a measurement value observed at one modifier level to the value under a reference modifier level that yields the same conditional prognostic quantity.
The basic operation is to equate a monotone conditional prognostic score and invert the reference-side curve.
In observational data, however, treatment may be selected according to the modifier and the measurement, so a direct equate-and-invert procedure may reflect differences in treatment assignment as well as the prognostic meaning of the measurement itself.
We therefore define, in addition to a direct approach, a policy-standardized mapping that standardizes treatment under a common policy using the g-formula.
We also introduce an origin-referenced mapping that compares changes in prognostic score from a common anchor, thereby separating modifier-specific prognostic levels from modifier dependence in the measurement--prognosis relationship.
Using Cox models, we develop linear and flexible inversion estimators with analytic and bootstrap confidence bands.
Simulations evaluate finite-sample performance and illustrate how treatment standardization and origin referencing clarify what is being mapped.
\end{abstract}


\section{Introduction}
% ===========================================================================

In clinical practice, a value measured on the same physical scale may have different prognostic implications depending on the patient's clinical background.
This issue is well illustrated by tumor size in childhood soft-tissue sarcoma: because the prognostic weight of a given diameter depends on the size of the child's body, the same absolute tumor size is more advanced in a smaller (younger) patient than in a larger (older) one, so that staging by raw diameter can misrepresent severity. 
\citet{Ferrari2009-wk} accordingly interpreted tumor size relative to body surface area and derived a size conversion across children of different body sizes.

The same structure appears across medicine. In lung adenocarcinoma, the prognostic information conveyed by total tumor size differs from that conveyed by the invasive component \citep{Kameda2018-uc}. 
Outside oncology, for example, creatinine-based estimates of the glomerular filtration rate are influenced by non-GFR determinants such as muscle mass \citep{Hsu2021-lp}. 
These examples share a common feature: a measurement scale calibrated in one subpopulation does not necessarily preserve the same clinical meaning in another.

These examples can be viewed as cases in which patient characteristics modify the clinical or prognostic meaning of a measurement. 
Although this phenomenon is related to effect modification and interaction \citep{VanderWeele2009} and to calibration in prognostic modeling \citep{Harrell1996, VanCalster2019}, the present problem is more specific. It asks how a measurement observed under one modifier value should be translated to a reference modifier value so that the translated value carries the same prognostic meaning. 
To formalize this idea, we define the \emph{prognosis-equivalent mapping}, which maps a measurement $X=x$ observed under modifier $M=m$ to the value under a reference modifier $m_{\rm ref}$ that carries the same prognostic meaning.

The equate-and-invert construction underlying this mapping has precedents in several areas. 
In bioassay, relative potency is defined by equating biological responses and inverting a dose--response model \citep{Finney1978-assay, DinseUmbach2011-rp}. 
In epidemiology, risk and rate advancement periods solve for the shift in age that gives the same risk or hazard between exposed and reference groups \citep{Brenner1993-rap}. 
Similarly, vascular and biological ages invert risk or mortality models to express an individual's prognosis as the age of a reference population with the same risk \citep{DAgostino2008-fram, Levine2018-phenoage}. 
These examples share the basic operation of equating an outcome-related quantity and solving on a reference scale. 
The present problem differs in its target: it maps measurement values across levels of a patient-level modifier by equating their conditional time-to-event prognosis.

Several statistical frameworks are also adjacent but answer different questions. 
Conditional and counterfactual distributional methods compare distributions across groups \citep{Chernozhukov2013-cf}, whereas the present target equates conditional prognosis. 
Measurement-error methods, including regression calibration and SIMEX, correct covariate-error bias in regression parameters \citep{Prentice1982-qn, Cook1994-sa}, rather than estimating a between-modifier mapping. Varying-coefficient models estimate how a covariate effect changes with another variable \citep{Hastie1993-zj}, but do not by themselves define a reference-scale value carrying equivalent prognosis. 

Clinical examples of outcome-based recalibration exist, but they have largely remained ad hoc, problem-specific constructions. 
The tumor-size conversion of \citet{Ferrari2009-wk}, based on iso-risk curves from a Cox interaction model, is a motivating special case of the prognosis-equivalent mapping formalized here.
What has been missing is a general estimand, together with its causal interpretation, identification assumptions, and inferential procedures.
This is particularly important when prognosis depends not only on the measurement and modifier themselves but also on treatment decisions that may vary with both the measurement and the modifier.
In that setting, the target mapping depends on whether treatment-mediated heterogeneity is retained as part of observed clinical practice or standardized under a common policy.

This article aims to fill that gap with a three-layered contribution. 
First, we define the prognosis-equivalent mapping through a general prognostic quantity and separate modifier-specific prognostic levels from modifier dependence in the measurement--prognosis relationship (Section~\ref{sec:formulation}).
Second, we give these mappings an explicit causal interpretation using causal diagrams in which the modifier may affect prognosis through a modifier-specific prognostic level, modification of the measurement--prognosis relationship, and treatment choice.
This structure motivates two versions: the \emph{observed-practice mapping}, based on observed clinical practice, and the \emph{policy-standardized mapping}, which standardizes treatment choice under a common policy (Section~\ref{sec:causal}).
Third, under a Cox proportional hazards model, we develop estimation and inference procedures for right-censored time-to-event data, including a closed-form mapping under a linear interaction model, numerical inversion under a flexible model with monotone rearrangement, and a g-computation procedure for the policy-standardized mapping (Sections~\ref{sec:estimation} and \ref{sec:inference}). 
The article then presents a simulation study (Section~\ref{sec:simulation}) and concluding discussion (Section~\ref{sec:discussion}).

% ===========================================================================
\section{Problem formulation and estimation target}
\label{sec:formulation}
% ===========================================================================

\subsection{Setup}

We consider a setting in which the prognostic meaning of a continuous measurement $X$ varies with a modifier $M$, conditional on baseline covariates $Z$. 
The event time $T$ may be right-censored by a censoring time $C$, so we observe $Y=\min(T,C)$ and $\Delta=I(T\le C)$. 
Estimation under censoring is described in Section~\ref{sec:estimation}.

Let $\mathcal{Q}_t(x,m,z)$ denote a prognostic quantity at time $t$, with the conditional survival function as the default choice; its concrete observed-practice and policy-standardized forms are introduced in Section~\ref{sec:causal}. 
Throughout we maintain the \emph{prognostic-monotonicity} condition that $\mathcal{Q}_t(\cdot,m,z)$ is strictly monotone in $x$, in a common direction, for every $(m,z)$. 
This fixes the prognostic ordering of $x$ and, together with the existence conditions stated below, makes the mappings well defined. 
It is convenient to work on an unbounded scale, so let $\varphi$ be a known strictly monotone transformation and define the \emph{prognostic score}
\begin{gather}
\Psi_t^{\mathcal{Q}}(x,m,z)
:=
\varphi\left\{\mathcal{Q}_t(x,m,z)\right\}.
\end{gather}
A natural choice is the complementary log-log transformation $\varphi(q)=\log\{-\log(q)\}$, which gives the log-cumulative-hazard scale when $\mathcal{Q}_t$ is a survival probability and corresponds to the linear predictor scale of a Cox model \citep{Cox1972, AndersenGill1982}. 
Since $\varphi$ is strictly monotone, prognostic monotonicity carries over to $\Psi_t^{\mathcal{Q}}(\cdot,m,z)$. 

The goal is a mapping that sends $X=x$ observed under $M=m$ to the value carrying equivalent prognostic score under $M=m_{\rmref}$. We define the absolute and origin-referenced mappings on the generic quantity $\mathcal{Q}_t$ (Section~\ref{sec:peq-mappings}). 

\subsection{Prognosis-equivalent mappings}
\label{sec:peq-mappings}

Write $\mathcal{X}_{m,z}$ for a prespecified measurement domain within the conditional support of $X$ given $(M,Z)=(m,z)$, over which the prognostic score is evaluated and inversion is considered.
For a chosen reference modifier value $m_{\rm ref}$, $\mathcal{X}_{m_{\rm ref},z}$ denotes the corresponding reference-side domain in which the mapped value $L$ is sought.

The most direct mapping equates the full prognostic score.

\begin{definition}[Absolute prognosis-equivalent mapping]
\label{def:abs-map}

For fixed $t$, $(m,z)$, and $m_{\rm ref}$, define the \emph{supported domain} as the set of all $x\in\mathcal{X}_{m,z}$ satisfying
\begin{gather}
\Psi_t^{\mathcal Q}(x,m,z)
\in
\left\{
\Psi_t^{\mathcal Q}(u,m_{\rm ref},z)
:
u\in\mathcal{X}_{m_{\rm ref},z}
\right\}.
\end{gather}
For each $x$ in this supported domain, prognostic monotonicity ensures that the \emph{absolute prognosis-equivalent mapping} $L_t^{\mathcal Q}(x,m,z;m_{\rm ref})$ is the unique value
$L\in\mathcal{X}_{m_{\rm ref},z}$ satisfying
\begin{gather}
\Psi_t^{\mathcal Q}(L,m_{\rm ref},z)
=
\Psi_t^{\mathcal Q}(x,m,z).
\end{gather}

\end{definition}

Equivalently, $\mathcal{Q}_t(L,m_{\rm ref},z)=\mathcal{Q}_t(x,m,z)$: the absolute mapping returns the value under $M=m_{\rm ref}$ that carries the same total prognosis as $X=x$ under $M=m$, for the same $z$.

Because the absolute mapping equates the full prognostic score, its mapped value reflects two sources of prognostic difference: modifier groups may have different prognostic scores at a common measurement value, and the prognostic score may vary differently with $X$ across groups.
To distinguish these two sources, we express each modifier group’s prognostic score relative to a common measurement anchor.
For a prespecified anchor $x_0$, this is done by subtracting each group’s prognostic score at $X=x_0$.
The resulting excess score is zero at the anchor in every group and captures how the prognostic score changes as $X$ moves away from it.
Because the resulting excess score and mapping generally depend on $x_0$, the choice of anchor is part of the estimand.
In the motivating tumor-size setting, the origin $x_0=0$ is a natural choice because it represents the absence of the measured lesion or component across modifier groups.
Here, $X=0$ serves only as a prognostic reference point and does not imply zero mortality risk.

Using $x_0=0$, we define the excess prognostic score as
\begin{gather}
\Psi_t^{\mathcal Q,\circ}(x,m,z)
:=
\Psi_t^{\mathcal Q}(x,m,z)
-
\ell_t(m,z),
\end{gather}
where $\ell_t(m,z):=\Psi_t^{\mathcal Q}(0,m,z)$ is the \emph{modifier-specific prognostic level} at the origin.
Thus
\begin{gather}
\Psi_t^{\mathcal Q}(x,m,z)
=
\ell_t(m,z)
+
\Psi_t^{\mathcal Q,\circ}(x,m,z).
\end{gather}
Differences in $\ell_t(m,z)$ across modifier values represent modifier-associated differences in prognostic level at the origin, for fixed $z$.
Dependence of the excess score $\Psi_t^{\mathcal Q,\circ}$ on $m$ represents modification of the relationship between $X$ and prognosis relative to that origin.
These level and interaction interpretations are defined relative to the chosen prognostic-score scale and anchor.

Equating the excess prognostic scores defines the second mapping.

\begin{definition}[Origin-referenced prognosis-equivalent mapping]
\label{def:oref-map}

For fixed $t$, $(m,z)$, and $m_{\rm ref}$, suppose that $0 \in\mathcal X_{m,z}\cap\mathcal X_{m_{\rm ref},z},$ and that both anchor scores $\Psi_t^{\mathcal Q}(0,m,z)$ and $\Psi_t^{\mathcal Q}(0,m_{\rm ref},z)$ are finite.

Define the \emph{supported domain} as the set of all $x\in\mathcal{X}_{m,z}$ satisfying
\begin{gather}
\Psi_t^{\mathcal Q,\circ}(x,m,z)
\in
\left\{
\Psi_t^{\mathcal Q,\circ}(u,m_{\rm ref},z)
:
u\in\mathcal X_{m_{\rm ref},z}
\right\}.
\end{gather}
For each $x$ in this supported domain, prognostic monotonicity ensures that the \emph{origin-referenced prognosis-equivalent mapping} $L_t^{\mathcal Q,\circ}(x,m,z;m_{\rm ref})$ is the unique value $L\in\mathcal X_{m_{\rm ref},z}$ satisfying
\begin{gather}
\Psi_t^{\mathcal Q,\circ}(L,m_{\rm ref},z)
=
\Psi_t^{\mathcal Q,\circ}(x,m,z).
\end{gather}

\end{definition}

The origin-referenced mapping therefore equates the change in prognostic score from $X=0$, rather than the full prognostic score. 
Unlike the absolute mapping, it generally depends on the prespecified transformation $\varphi$ because centering is performed after the transformation. 
Because the excess score is zero at the origin in both modifier groups, prognostic monotonicity implies the fixed-point property $L_t^{\mathcal Q,\circ}(0,m,z;m_{\rm ref})=0$.

The following proposition gives the condition under which the two mappings coincide.
\begin{proposition}[Equivalence of the absolute and origin-referenced mappings]
\label{prop:recovery}
Fix $t$, $(m,z)$, and $m_{\rm ref}$, and suppose that the supported
domains of the absolute and origin-referenced mappings have a
nonempty intersection.
If
\begin{gather}
\ell_t(m,z)
=
\ell_t(m_{\rm ref},z),
\end{gather}
then the two supported domains coincide and
\begin{gather}
L_t^{\mathcal Q}(x,m,z;m_{\rm ref})
=
L_t^{\mathcal Q,\circ}(x,m,z;m_{\rm ref})
\end{gather}
at every $x$ in that domain.
Conversely, if the two mappings agree at any point in their common
supported domain, then
$\ell_t(m,z)=\ell_t(m_{\rm ref},z)$.
Consequently, the two mappings coincide throughout their supported
domains for all $(m,z)$ under consideration if and only if the
prognostic level is modifier-invariant.
% Under prognostic monotonicity, suppose that, for each $(m,z)$ under consideration, the supported domains of the absolute and origin-referenced mappings have a nonempty intersection.
% The two mappings coincide at every $x$ in their common supported domain, for all $(m,z)$,
% \begin{gather}
% L_t^{\mathcal{Q}}(x,m,z;m_{\rmref})
% =
% L_t^{\mathcal{Q},\circ}(x,m,z;m_{\rmref}),
% \end{gather}
% if and only if the prognostic level is modifier-invariant,
% \begin{gather}
% \ell_t(m,z)=\ell_t(m_{\rmref},z)
% \end{gather}
% for all $(m,z)$.
% Otherwise, wherever the relevant inverse exists, the absolute mapping additionally absorbs the level difference $\ell_t(m,z)-\ell_t(m_{\rmref},z)$ through the inverse score under $m_{\rmref}$.
\end{proposition}

\noindent The result follows directly from the two equations and monotonicity; see \ref{appA}.

A linear interaction model illustrates the decomposition. Suppose the score is additive with an $m$-free baseline,
\begin{gather}
\Psi_t^{\mathcal{Q}}(x,m,z)=b_t+\eta(x,m,z),
\qquad
\eta(x,m,z)=\beta_X x+\beta_M m+\beta_{XM}mx+\gamma^\top z;
\end{gather}
the complementary log-log score of a Cox model has this form, with $b_t=\log\Lambda_0(t)$, and this model is adopted in Section~\ref{sec:monotone-inversion}. The baseline cancels in the excess score, giving
\begin{gather}
\ell_t(m,z)=b_t+\beta_M m+\gamma^\top z,
\qquad
\Psi_t^{\mathcal{Q},\circ}(x,m,z)
=
(\beta_X+\beta_{XM}m)x.
\end{gather}
Therefore
\begin{gather}
L_t^{\mathcal{Q},\circ}(x,m,z;m_{\rmref})
=
\frac{\beta_X+\beta_{XM}m}
{\beta_X+\beta_{XM}m_{\rmref}}\,x,
\end{gather}
which depends on the $M\times X$ interaction, whereas the absolute mapping additionally includes the level offset $\beta_M(m-m_{\rmref})/(\beta_X+\beta_{XM}m_{\rmref})$.
Both mappings are also independent of $t$ in this model, since the time-dependent baseline term $b_t$ cancels from the defining equations.
Thus, in this model, the two mappings coincide exactly when $\beta_M=0$.

% ===========================================================================
\section{Causal structure, identification, and interpretation}
\label{sec:causal}
% ===========================================================================

The mappings in Section~\ref{sec:formulation} are defined through conditional prognosis, but treatment decisions complicate their interpretation because treatment may depend on $X$, $M$, and $Z$ and may itself affect the outcome.
The causal formulation is not intended to interpret changes in $X$ or $M$ as interventions; rather, it distinguishes prognosis under observed treatment practice from prognosis standardized to a common treatment policy and states the conditions identifying the latter.

\subsection{Causal diagrams and target quantities}

Using causal diagrams \citep{Pearl1995, Greenland1999}, Figure~\ref{fig:dag} summarizes the assumed data-generating structure.
We then define the corresponding observed-practice and policy-standardized prognostic quantities.

\begin{figure}[htbp]
    \centering
    \begin{minipage}[t]{0.48\textwidth}
        \centering
        \begin{tikzpicture}[
            node distance=0.9cm and 1.2cm,
            every node/.style={draw, circle, minimum size=0.8cm, inner sep=0pt, font=\small},
            arr/.style={-{Stealth[length=6pt]}, thick},
            mod/.style={-{Stealth[length=5pt]}, thick, dashed},
            bidir/.style={{Stealth[length=6pt]}-{Stealth[length=6pt]}, thick, dashed}
        ]
        \node (M)              {$M$};
        \node (Z) [above=of M] {$Z$};
        \node (X) [below=of M] {$X$};
        \node (A) [right=of M] {$A$};
        \node (T) [right=of A] {$T$};
        \draw[arr] (M) to[bend left=30] (T);
        \draw[arr] (M) -- (A);
        \draw[arr] (X) to[bend right=45] (T);
        \draw[arr] (X) -- (A);
        \draw[arr] (Z) -- (A);
        \draw[arr] (Z) to[bend left=35] (T);
        \draw[arr] (A) -- (T);
        \draw[bidir] (Z) to[bend right=35] (X);
        \draw[bidir] (M) -- (Z);
        \draw[bidir] (M) -- (X);
        \end{tikzpicture}

        {\small (a) Effect modification: $M$ moderates $X\to T$.}
    \end{minipage}
    \hfill
    \begin{minipage}[t]{0.48\textwidth}
        \centering
        \begin{tikzpicture}[
            node distance=0.9cm and 1.2cm,
            every node/.style={draw, circle, minimum size=0.8cm, inner sep=0pt, font=\small},
            latent/.style={dashed},
            arr/.style={-{Stealth[length=6pt]}, thick},
            mod/.style={-{Stealth[length=5pt]}, thick, dashed},
            bidir/.style={{Stealth[length=6pt]}-{Stealth[length=6pt]}, thick, dashed}
        ]
        \node (M)                         {$M$};
        \node (Z)     [above=of M]        {$Z$};
        \node (A)     [right=of M]        {$A$};
        \node (T)     [right=of A]        {$T$};
        \node (Xstar) [below=of M, latent] {$X^*$};
        \node (X)     [right=of Xstar]        {$X$};
        \draw[arr] (Xstar) -- (X);
        \draw[arr] (Xstar) to[bend right=60] (T);
        \draw[arr] (M) to[bend left=30] (T);
        \draw[arr] (M) -- (X);
        \draw[arr] (M) -- (A);
        \draw[arr] (X) -- (A);
        \draw[arr] (Z) -- (A);
        \draw[arr] (Z) to[bend left=35] (T);
        \draw[arr] (A) -- (T);
        \draw[bidir] (Z) to[bend right=35] (Xstar);
        \draw[bidir] (M) -- (Z);
        \draw[bidir] (M) -- (Xstar);
        \end{tikzpicture}

        {\small (b) Measurement special case: $X=g_M(X^*)$, $X^*\to T$.}
    \end{minipage}
    \caption{Two readings of the same observed data for the prognosis-equivalent mapping problem. 
    Solid arrows indicate directed causal effects, whereas dashed bidirectional links indicate possible unobserved common causes between baseline variables.
    \textbf{(a)} Effect-modification reading, with no latent quantity: $M$ acts on prognosis through a modifier-specific \emph{level} ($M\to T$) and on treatment choice ($M\to A$), and the observed measurement carries the prognostic effect $X\to T$.
    The \emph{$M\times X$ interaction}, by which the prognostic effect of $X$ depends on $M$, is a property of the functional form of the joint effect of $M$ and $X$ on $T$ rather than a separate edge. 
    \textbf{(b)} Measurement special case of Proposition~\ref{prop:units}: the prognostic effect instead runs from a latent common-unit quantity $X^*$ ($X^*\to T$), and the observed $X$ is a modifier-dependent measurement of $X^*$, $X=g_M(X^*)$, encoded by the edges $X^*\to X$ and $M\to X$.}
    \label{fig:dag}
    \vspace{3mm}
    % \raggedright\noindent\textit{Alt text:} Two causal diagrams sharing nodes M (modifier), Z (pretreatment covariates), X (observed measurement), A (treatment), and T (event time). In panel (a) directed edges run from M to T, X to T, M to A, X to A, Z to A, Z to T, and A to T, and dashed bidirectional links connect M with Z, M with X, and Z with X; the M-by-X interaction is not a drawn edge but the dependence of the X-to-T effect on M. Panel (b) routes the prognostic effect through a latent node X-star instead of through X: directed edges run from X-star to T, from X-star to X, and from M to X, so that X is a modifier-dependent measurement of X-star (the measurement map depends on M), together with M to T, M to A, X to A, Z to A, Z to T, and A to T, and dashed bidirectional links connect M with Z, M with X-star, and Z with X-star. Panel (b) is the measurement special case in which the observed X is a modifier-specific reading of the latent X-star.
\end{figure}
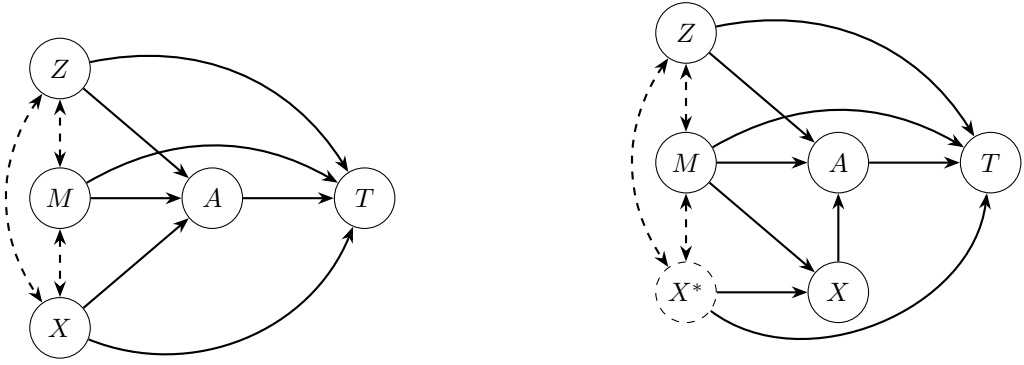

Panel~(a) shows the basic causal diagram for the observed variables.
The modifier $M$ may affect prognosis through its own prognostic level ($M\to T$), through an interaction with $X$, and through treatment choice ($M\to A\to T$). 
Because treatment is chosen based on $X$, $M$, and $Z$, the prognostic meaning of $X$ in the observed data depends on the prevailing treatment policy. 
This motivates two target quantities for $\mathcal{Q}_t$ in Section~\ref{sec:formulation}.

The \emph{observed-practice (OP)} quantity is the conditional survival under the treatment process actually operating in the data,
\begin{gather}
    Q_t^{\obs}(x,m,z)
    :=
    P(T>t \mid X=x,M=m,Z=z).
\end{gather}
It reflects all three pathways by which $M$ acts in panel~(a), including the treatment-mediated one ($M\to A\to T$), and so represents the composite prognostic meaning of $X$ under current practice. 
Setting $\mathcal{Q}_t=Q_t^{\obs}$ in Definitions~\ref{def:abs-map}--\ref{def:oref-map} yields the absolute and origin-referenced OP mappings $L_t^{\obs}$ and $L_t^{\obs,\circ}$.

The \emph{policy-standardized (PS)} quantity replaces the prevailing treatment mechanism with a common policy $\pi$.
Let $T^a$ denote the potential event time under the intervention $A=a$, and let $T^\pi$ denote the counterfactual event time under policy $A\sim\pi$; define
\begin{gather}
    Q_t^\pi(x,m,z)
    :=
    P(T^\pi>t \mid X=x,M=m,Z=z).
\end{gather}
Setting $\mathcal{Q}_t=Q_t^\pi$ yields the absolute and origin-referenced PS mappings $L_t^\pi$ and $L_t^{\pi,\circ}$. 
The common policy $\pi$ may be a fixed regime, such as assigning $A=a$ to everyone, or a conditional rule that assigns treatment with probability $\pi(a\mid\cdot)$.
The variables allowed to enter $\pi$ determine which treatment-mediated heterogeneity is retained or standardized away \citep{HernanTaubman2008, Young2014}.

Choosing an $M$-independent policy $\pi(a\mid x,z)$ standardizes modifier-specific treatment differences operating through $M\to A\to T$ at fixed $(x,z)$ and imposes a common relationship between the recorded measurement and treatment choice across modifier groups.
It therefore removes the influence of clinical behavior in which the same recorded value of $X$ is interpreted differently for treatment according to $M$, and defines prognostic equivalence when all groups are managed using the same function of recorded $X$ and $Z$.

If instead a fixed regime $A=a$ or a policy $\pi(a\mid z)$ independent of both $M$ and $X$ is chosen, treatment assignment is (conditionally) independent of $(M,X)$.
Such a policy standardizes prognostic differences arising through modifier- or measurement-dependent treatment selection, including the pathways $M\to A\to T$, $X\to A\to T$, and $M\to X\to A\to T$.
The resulting mapping compares prognosis under a common treatment regime or treatment distribution.
It does not, however, remove the prognostic effect of treatment itself or effect modification of treatment by $M$ or $X$.

Panel~(b) adds a measurement interpretation. 
The observed measurement is a modifier-dependent reading, $X=g_M(X^*)$, of a latent common-unit quantity $X^*$.
Under this model, modifier dependence in the observed $X$--prognosis relationship may arise from the measurement map $g_M$, the latent-scale prognostic relationship, or both.
This interpretation is natural for examples such as invasive-component tumor size and creatinine-based eGFR. 
In contrast, the childhood-sarcoma example follows the panel-(a) effect-modification interpretation, in which tumor size is not viewed as a modifier-dependent measurement of a latent common-unit quantity.
Section~\ref{sec:recovery} further discusses this measurement-based interpretation.

\subsection{Identification of the policy-standardized mapping}
\label{sec:identification}

To evaluate the PS mapping from observed data, $Q_t^\pi(x,m,z)$ must be identifiable. 
For a baseline point treatment, we assume the standard conditions for g-formula identification \citep{Robins1986, HernanRobins2020}: consistency of the treatment-specific counterfactual event times $T^a$, conditional exchangeability $T^a \ci A\mid X,M,Z$ for treatment values assigned positive probability under $\pi$ at the target $(x,m,z)$, treatment positivity at those values of $(X,M,Z)$, and independent censoring conditional on at least $(X,M,Z,A)$ with positive probability of remaining uncensored up to time $t$.
These conditions are stated formally as Assumptions~(A1)--(A4) in \ref{appA}.

\begin{proposition}[Identification of $Q_t^\pi$]
\label{prop:identification}
Under Assumptions (A1)--(A4), the PS prognostic quantity $Q_t^\pi(x,m,z)$ is identified from the observed-data distribution by
\begin{gather}
    Q_t^\pi(x,m,z)
    =
    \sum_a
    P(T>t \mid X=x,M=m,Z=z,A=a)\,\pi(a\mid x,m,z).
    \label{eq:gcomp-identification}
\end{gather}
Consequently, suppose that the g-formula quantities are identified at the source point and throughout the reference domain $\mathcal{X}_{m_{\rmref},z}$.
When the reference side $Q_t^\pi(\cdot,m_{\rmref},z)$ is strictly monotone in $x$ and $Q_t^\pi(x,m,z)$ belongs to its range over $\mathcal{X}_{m_{\rmref},z}$, the unique PS mapping $L_t^\pi(x,m,z;m_{\rmref})$ is identified from the observed-data distribution.
The analogous conclusion holds for $L_t^{\pi,\circ}$ when the anchor scores are identified and finite and the source excess score belongs to the range of the reference excess score.
\end{proposition}
\noindent The proof is given in \ref{appA}. 

In practice, the prognostic function $Q_t^\pi$ must be consistently estimated, requiring either a correctly specified parametric model or a sufficiently flexible estimation approach.

\subsection{Measurement special case: a reference-unit conversion}
\label{sec:recovery}

To gain further insight into the mappings, we consider the structured measurement setting introduced in panel~(b) of Figure~\ref{fig:dag}, in which $X$ is a modifier-specific measurement of an underlying quantity on a common scale as $X=g_M(X^*)$.
Under suitable additional conditions, an appropriate mapping coincides with $L(x)=g_{m_{\rmref}}\!\left\{g_m^{-1}(x)\right\}$, which converts measurements to the corresponding reference-group scale and is therefore termed the \emph{reference-unit conversion}.

In this subsection, we allow for treatment-mediated heterogeneity and focus on the PS mapping under either a fixed treatment or a policy $\pi(a\mid z)$.
These interventions remove treatment-assignment pathways through both $M$ and $X$ and therefore provide the clearest basis for a reference-unit interpretation.
In the absence of treatment-mediated heterogeneity, the corresponding OP mapping may also admit the same interpretation.

For the chosen policy $\pi$, define the latent-scale prognostic score by
\begin{gather}
\Psi_t^{\pi,\ast}(x^*,m,z)
:=
\Psi_t^\pi\!\left(g_m(x^*),m,z\right),
\end{gather}
and its excess relative to the latent origin by
\begin{gather}
\Psi_t^{\pi,\ast,\circ}(x^*,m,z)
:=
\Psi_t^{\pi,\ast}(x^*,m,z)
-
\Psi_t^{\pi,\ast}(0,m,z).
\end{gather}
Then, we consider the following conditions.
\begin{enumerate}
\item[\textnormal{(G1)}] {\bf Invertibility.}
For each modifier value $m$, the measurement map $g_m(x^*)$ is strictly increasing in $x^*$, hence invertible.

\item[\textnormal{(G2)}] {\bf Origin preservation.}
$g_m(0)=0$ for every modifier value $m$, so $X=0$ corresponds to $X^*=0$.

\item[\textnormal{(G3)}] {\bf Latent excess-score invariance.}
There exists a function $\widetilde\Psi_t^{\pi,\circ}(x^*,z)$ that does not depend on $m$ such that
\begin{gather}
\Psi_t^{\pi,\ast,\circ}(x^*,m,z)
=
\widetilde\Psi_t^{\pi,\circ}(x^*,z)
\end{gather}
for every $(x^*,m,z)$ under consideration.

\item[\textnormal{(G4)}] {\bf No residual modifier effect.}
$T^\pi \ci M \mid X^*,Z$. 
\end{enumerate}

\begin{proposition}[Reference-unit conversion]
\label{prop:units}
Under the measurement model $X=g_M(X^*)$ and the maintained prognostic monotonicity:
\begin{enumerate}
\item[\textnormal{(i)}]
Under (G1)--(G3), the origin-referenced PS mapping is the reference-unit conversion
\begin{gather}
L_t^{\pi,\circ}(x,m,z;m_{\rmref})
=
g_{m_{\rmref}}\!\left\{g_m^{-1}(x)\right\}.
\end{gather}
\item[\textnormal{(ii)}]
Under (G1) and (G4), the absolute PS mapping equals the same conversion,
\[
L_t^\pi(x,m,z;m_{\rmref})
=
g_{m_{\rmref}}\!\left\{g_m^{-1}(x)\right\}.
\]
\end{enumerate}
\end{proposition}

\noindent The proof is given in \ref{appA}.

Condition (G2) connects the observed anchor $X=0$ used by the origin-referenced mapping to the latent origin $X^*=0$ used in (G3).
Under (G3), modifier-specific prognostic levels are allowed, but the change in prognosis from the latent origin is common across modifier groups.
Condition (G4) is stronger: it requires the full counterfactual outcome distribution at fixed $(X^*,Z)$ to be invariant across modifier groups and therefore implies (G3).
Under (G4), the absolute PS mapping also reduces to the reference-unit conversion and does not require (G2), because it does not use the observed origin as an anchor.
These interpretations require substantive latent measurement assumptions and cannot be established from the observed-data distribution alone.

% ===========================================================================
\section{Estimation under Cox proportional hazards model}
\label{sec:estimation}
% ===========================================================================

We describe plug-in estimation of the prognosis-equivalent mappings under Cox proportional hazards model \citep{Cox1972, AndersenGill1982}. 
The observed data are $(Y,\Delta,X,M,Z)$ for the observed-practice mapping and $(Y,\Delta,X,M,Z,A)$ for the policy-standardized mapping. 
Let $\eta(\cdot)$ be a prognostic score that does not depend on time $t$, and write the conditional hazard function in the Cox form $\lambda(t\mid \cdot)=\lambda_0(t)\exp\{\eta(\cdot)\}$, where $\lambda_0(t)$ is the baseline hazard function. 
Then the conditional survival function is given by $S(t\mid \cdot) = \exp\left[-\Lambda_0(t)\exp\{\eta(\cdot)\}\right]$, where $\Lambda_0(t)$ is the baseline cumulative hazard function.

The Cox model is used here as a working model for the prognostic score. When correctly specified, the plug-in estimators below are consistent for the corresponding mappings.

\subsection{Plug-in mapping estimators}
\label{sec:plugin-estimators}

For the observed-practice mapping, we fit a Cox model to $(Y,\Delta,X,M,Z)$ and obtain the fitted score $\hat\eta(x,m,z)$. 
For the OP Cox fit, censoring is assumed independent of $T$ conditional on $(X,M,Z)$.
Since the same baseline cumulative hazard $\Lambda_0(t)$ appears on both sides of the prognostic-equivalence equation, equating the OP survival is equivalent to equating the Cox score:
\begin{gather}
\eta(L,m_{\rmref},z)=\eta(x,m,z).
\label{eq:eta-equate}
\end{gather}
Thus, under a single Cox score, the OP mapping does not depend on $t$. 
Its absolute plug-in estimator $\hat L^\obs(x,m,z;m_{\rmref})$ solves $\hat\eta(L,m_{\rmref},z)=\hat\eta(x,m,z)$, and the origin-referenced estimator solves $\hat\eta(L,m_{\rmref},z)-\hat\eta(0,m_{\rmref},z) = \hat\eta(x,m,z)-\hat\eta(0,m,z)$.

For the PS mapping, we fit a Cox model that includes treatment $A$, to $(Y,\Delta,X,M,Z,A)$, and write the fitted score as $\hat\eta_Y(x,m,z,a)$. 
The treatment-specific survival is estimated by 
\begin{gather}
\widehat S(t\mid x,m,z,a)
=
\exp\left[-\hat\Lambda_0(t)\exp\{\hat\eta_Y(x,m,z,a)\}\right],
\end{gather}
and the common-policy prognostic quantity is estimated by
\begin{gather}
\widehat Q_t^\pi(x,m,z) = \sum_a \widehat S(t\mid x,m,z,a)\,\pi(a\mid x,m,z).
\end{gather}

Here, $\pi$ is a prespecified target policy that forms part of the definition of the PS estimand and is treated as fixed throughout estimation.
The absolute PS estimator $\hat L_t^\pi(x,m,z;m_{\rmref})$ solves $\widehat Q_t^\pi(L,m_{\rmref},z) = \widehat Q_t^\pi(x,m,z)$, whereas the origin-referenced estimator $\hat L_t^{\pi, \circ}$ solves $\widehat\Psi_t^{\pi,\circ}(L,m_{\rmref},z) = \widehat\Psi_t^{\pi,\circ}(x,m,z)$, with
\begin{gather}
\widehat\Psi_t^{\pi,\circ}(x,m,z)
:=
\log\{-\log\widehat Q_t^\pi(x,m,z)\}
-
\log\{-\log\widehat Q_t^\pi(0,m,z)\}.
\end{gather}
Unlike the OP mapping, the PS mapping generally depends on $t$, because $\widehat Q_t^\pi$ is a mixture of treatment-specific survival functions and need not reduce to a single Cox score. 
If $\pi=\delta_{a_0}$ assigns the same treatment $a_0$ to everyone, the mixture collapses to a single fitted Cox survival curve, and the mapping again reduces to equating $\hat\eta_Y(\cdot,\cdot,\cdot,a_0)$.

\subsection{Score specification and numerical inversion}
\label{sec:monotone-inversion}

The Cox score may be specified either parametrically or flexibly. 
A useful parametric illustration is the linear interaction score
\begin{gather}
\eta(x,m,z)
=
\beta_X x+\beta_M m+\beta_{XM}mx+\gamma^\top z.
\label{eq:linear-interaction-score}
\end{gather}
In this case, the OP absolute mapping has the closed form
\begin{gather}
L^\obs(x,m,z;m_{\rmref})
=
\frac{
\beta_M(m-m_{\rmref})+(\beta_X+\beta_{XM}m)x
}{
\beta_X+\beta_{XM}m_{\rmref}
},
\end{gather}
whereas the origin-referenced mapping is
\begin{gather}
L^{\obs,\circ}(x,m,z;m_{\rmref})
=
\frac{\beta_X+\beta_{XM}m}
{\beta_X+\beta_{XM}m_{\rmref}}\,x.
\end{gather}
Thus the absolute mapping contains both the level shift and the modifier-specific slope change, while the origin-referenced mapping removes the level shift and retains only the excess prognostic contribution of $X$ relative to the origin.
Order preservation requires the reference and target slopes, $\beta_X+\beta_{XM}m_{\rmref}$ and $\beta_X+\beta_{XM}m$, to have the same sign.

In real data, the score need not be linear in $X$ or in its interaction with $M$. 
We therefore allow $\eta$ and $\eta_Y$ to be flexible Cox scores, estimated by maximizing the Cox partial likelihood. 
The mappings are then obtained from the same equating equations as above, using numerical inversion over the supported reference range.

The mapping is evaluated only when the target value lies within the supported range of the reference curve and the reference-side curve is prognostically monotone. 
If the underlying curve is genuinely non-monotone, a global mapping is not well defined. When the underlying curve is monotone, flexible fitted curves may still show minor finite-sample departures from monotonicity.
We therefore check monotonicity on the evaluation grid and, when needed, apply grid-based monotone rearrangement before inversion \citep{Chernozhukov2009-rr}.

For origin-referenced mappings with flexible scores, the anchor itself must also be well supported. 
If the usual anchor $X=0$ lies outside the observed support, its score is not nonparametrically identified and becomes a model-dependent extrapolation; near the support boundary it may also be unstable.
We therefore do not recommend flexible origin-referenced mappings in such settings, unless a clinically meaningful interior anchor is available.

\subsection{Consistency of plug-in inversion}
\label{sec:plugin-consistency}

We state sufficient conditions for consistency of the plug-in inverse estimators under correct specification of the Cox models.
Let $\mathcal{D}$ be a compact set containing all source, mapped reference, and, where applicable, anchor arguments required for inversion.

For the absolute OP mapping, suppose
\begin{gather}
\sup_{(x,m,z)\in\mathcal{D}}
\left|
\hat\eta(x,m,z)-\eta(x,m,z)
\right|
\overset{p}{\longrightarrow}0,
\end{gather}
and suppose that the reference curve $\eta(\cdot,m_{\rmref},z)$ is continuous and strictly monotone on the relevant compact interval.
Then
\begin{gather}
\sup_{(x,m,z)\in\mathcal{D}}
\left|
\hat L^\obs(x,m,z;m_{\rmref})
-
L^\obs(x,m,z;m_{\rmref})
\right|
\overset{p}{\longrightarrow}0.
\end{gather}

For the absolute PS mapping, fix $t$ and suppose
\begin{gather}
\sup_{(x,m,z)\in\mathcal{D}}
\left|
\widehat Q_t^\pi(x,m,z)-Q_t^\pi(x,m,z)
\right|
\overset{p}{\longrightarrow}0,
\end{gather}
and suppose that the reference curve $Q_t^\pi(\cdot,m_{\rmref},z)$ is continuous and strictly monotone on the relevant compact interval.
Then
\begin{gather}
\sup_{(x,m,z)\in\mathcal{D}}
\left|
\hat L_t^\pi(x,m,z;m_{\rmref})
-
L_t^\pi(x,m,z;m_{\rmref})
\right|
\overset{p}{\longrightarrow}0.
\end{gather}

For the origin-referenced mappings, additionally suppose that the anchor $x=0$ belongs to the relevant source and reference domains and that the corresponding anchor scores are finite.
For the PS mapping, also assume that $Q_t^\pi$ is bounded away from $0$ and $1$ over the relevant source, reference, and anchor ranges.
Uniform consistency of $\hat\eta$ and $\widehat Q_t^\pi$ then implies uniform consistency of the corresponding origin-centered scores, on the Cox-score and complementary log-log scales, respectively.
Under the same assumptions, the analogous consistency results follow for both $\hat L^{\obs,\circ}$ and $\hat L_t^{\pi,\circ}$.
Monotone rearrangement in the maintained direction preserves these consistency conclusions.
Details are given in \ref{appA}.

% ===========================================================================
\section{Inference: confidence bands}
\label{sec:inference}
% ===========================================================================

We construct pointwise intervals and simultaneous confidence bands for the estimators in Section~\ref{sec:estimation}. 
The nonparametric bootstrap \citep{Efron1994-xd} is used for flexible-score and dynamic-policy estimators, whereas the linear OP and linear fixed-treatment PS estimators use the analytic delta method described in \ref{appB}.

In the nonparametric bootstrap, we resample individuals with replacement and, on each bootstrap sample, re-run all estimation stages. 
This gives bootstrap curves $\hat{L}^{(b)}$ on the evaluation grid for $b=1,\ldots,B$.

Pointwise confidence intervals use the percentile method.
Writing $\alpha$ for the noncoverage probability, at a fixed $x_0$ the interval is the $\alpha/2$ and $1-\alpha/2$ quantiles of $\hat{L}^{(1)}(x_0),\ldots,\hat{L}^{(B)}(x_0)$.

For simultaneous inference over $x_1<\cdots<x_K$, let $Q_p(\cdot)$ denote the pointwise $p$-quantile of the bootstrap replicates. 
For a common tail probability $\zeta$, define
\begin{gather}
    C_k
    =
    \Bigl[\, Q_{\zeta}\bigl(\hat{L}^{(b)}(x_k)\bigr),\
              Q_{1-\zeta}\bigl(\hat{L}^{(b)}(x_k)\bigr) \,\Bigr],
    \qquad k = 1,\dots,K.
    \label{eq:sim-band}
\end{gather}
We calibrate $\zeta$ as the largest value for which at least $1-\alpha$ of the bootstrap curves fall inside the percentile rectangle $\prod_{k=1}^K C_k$, giving the narrowest percentile band with simultaneous bootstrap coverage \citep{MontielOleaPlagborgMoller2019}.

When the prognostic score $\eta$ based on a linear interaction model is directly equated, analytic confidence intervals based on the asymptotic normality of the regression parameter estimates are available \citep{Scheffe1953}.
Details on the analytic confidence intervals and related tests are provided in \ref{appB}.

% ===========================================================================
\section{Simulation studies}
\label{sec:simulation}
% ===========================================================================

\subsection{Design and scenarios}
\label{sec:sim-scenarios}

We evaluate the finite-sample performance of the estimation and inference methods described in Sections~\ref{sec:estimation}--\ref{sec:inference} using Monte Carlo simulation. 
The performance measures are the bias and RMSE of the mapping estimators, the pointwise and simultaneous $95\%$ coverage probabilities of the confidence bands, and the band width. 
The full data-generating process, including the modifier, covariate, and latent-measurement distributions, the deterministic measurement map $X=g_M(X^*)$, the Weibull baseline, the treatment and censoring models, and the parameter values, is provided in \ref{appC}. 
In the main text, we report only the features that distinguish the simulation scenarios.
All mappings are evaluated at $Z=0$ on the grid $x\in\{0.5,1.0,1.5,2.0,2.5\}$, and the modifier-dependent stretch uses $\delta=0.5$, which fixes the reference-unit slope at $1/(1+\delta)=2/3$.
Regular performance summaries are restricted to grid points whose target score lies in the interior of the reference score range; we call these points \emph{interior-supported} below.
This is stricter than mapping existence at an attainable boundary.
A prespecified operational fallback for lower-range failures is described below.

Each scenario has a \emph{without-treatment} version (-1) that evaluates the OP mapping and a \emph{with-treatment} version (-2) that evaluates the OP mapping under the natural treatment mechanism together with the PS mapping under the fixed regime assigning $A=1$ to everyone (PS-stc).
The with-treatment OP analysis does not standardize $A$ and therefore retains treatment-mediated differences induced by confounded treatment assignment.

We report the following three scenario families.

\emph{S1a--S1b (binary $M$; G1--G4 hold)} assess whether the mappings are correctly identified under two measurement maps.
In \emph{S1a} (the \emph{linear} map), $L_{\mathrm{true}}(x)=x/(1+\delta)=2x/3$ for both the absolute and origin-referenced mappings (Proposition~\ref{prop:units}).
This scenario evaluates the validity of the confidence bands, as well as the adequacy of confounding adjustment, in a setting where the target mapping is correctly specified. 
In \emph{S1b} (the \emph{nonlinear} map), the measurement map is a smooth concave hinge: $L_{\mathrm{true}}$ rises with near-unit slope near the origin and flattens to slope $1/10$ beyond a knot, passing through the origin and lying below $y=x$. 
In this scenario, the linear model is misspecified.

\emph{S2a--S2c (binary $M$; single violations)} introduce one violation at a time relative to the S1a map: a direct level effect of $M$ (violating G4), an $M\times X^*$ interaction (violating G3), or an origin shift (violating G2).
At supported points, the estimators still cover their identified estimands.
The key difference is how each estimand departs from the reference-unit conversion, and this departure occurs asymmetrically for the absolute and origin-referenced mappings (Proposition~\ref{prop:recovery}).
For the without-treatment OP and static-policy PS targets, under S2a the unconstrained absolute solution is offset by $\gamma_M/\beta_X^{*}\approx-0.71$, whereas the origin-referenced mapping recovers the reference-unit conversion.
Under S2c, the roles are reversed: the origin-referenced mapping is offset by $c_1/(1+\delta)\approx0.20$.
Under S2b, both mappings diverge from the reference-unit conversion.

\emph{S3a--S3b (continuous $M$; G1--G4 hold)} use a generalized reference value $m_{\rmref}$.
In \emph{S3a}, the prognostic score is linear in $X$ and in the $X\times M$ interaction, with no main effect of $M$: $\eta(x,m,z)=\beta_X^{*}(1-\tfrac{m}{3})x+\beta_Z z=\beta_X^{*}x-\tfrac{\beta_X^{*}}{3}\,mx+\beta_Z z$.
The additive covariate term $\beta_Z z$ ($Z\sim\mathcal{N}(0,1)$) is common to the source and reference curves and therefore cancels in the mapping.
The true mapping is $L_{\mathrm{true}}(x;m\to m_{\rmref})=x\,\kappa(m)/\kappa(m_{\rmref})$, where $\kappa(m)=1-m/3$.
Thus, the linear model with an $X\times M$ interaction is correctly specified, and we evaluate the generalized reference values $m_{\rmref}\in\{0,0.5,1\}$.
In \emph{S3b}, the prognostic score is linear in $m$ but has a smoothed hinge in $x$: $\eta(x,m,z)=\beta_X^{*}\{(1-m)x+m\,h(x)\}+\beta_Z z$, where $h$ is the S1b hinge.
Thus, the source map at $m=1$ equals that hinge.
We report results for $m_{\rmref}=0$, because mappings to non-origin reference values require unstable inversion of a nearly flat reference curve.

\subsection{Estimators, inference, and metrics}
\label{sec:sim-estimators}

Two learners are used to estimate the prognostic score. 
\textbf{LL} is a linear learner that includes $X$, $M$, and their $X\times M$ interaction. 
For the binary contrast $m=1\to m_{\rmref}=0$ under OP and the static policy, LL yields the closed-form estimators $\hat\phi=(\hat\beta_X+\hat\beta_{XM})/\hat\beta_X$ and $\hat\alpha=\hat\beta_M/\hat\beta_X$, together with the delta-method analytic band.
\textbf{SL} uses a natural-spline basis in $X$ (\texttt{splines::ns}, $\mathrm{df}=3$), together with its interaction with $M$.
For SL, we first apply monotone rearrangement to the estimated prognostic-score curve and then invert the curve numerically \citep{DurrlemanSimon1989, Gray1992}. 
Inference for SL is based on the bootstrap.

Policy standardization is implemented using the g-computation procedure in Section~\ref{sec:plugin-estimators}.
PS-stc is the policy-standardized mapping under the static regime $\pi=\delta_{1}$ that sets $A=1$ for every individual.
The g-computation then reduces to the fitted Cox model evaluated at $A=1$.
For inference, LL under OP and PS-stc uses the analytic delta-method bands, consisting of pointwise Wald intervals and the Gaussian plug-in sup-$t$ simultaneous band detailed in \ref{appB}.
These settings are evaluated at $n\in\{500,1000,2000\}$.
SL under OP or PS-stc uses the nonparametric bootstrap.
For these settings, pointwise intervals are percentile intervals, and simultaneous bands are calibrated percentile bands with common tail probability $\zeta$, as described in Section~\ref{sec:inference}.
They are evaluated at $n\in\{1000,2000\}$.

At each interior-supported grid point, we summarize the results over $R$ replications by reporting bias, RMSE, empirical SE, mean reported SE, pointwise coverage, and band width.
Simultaneous coverage is the proportion of replications in which the simultaneous band contains the true mapping jointly over the interior-supported grid points.
For LL, we additionally report the bias, RMSE, and coverage of the scalar parameters $\phi$ and $\alpha$, together with the rejection rate for $H_0:\phi=1$.

For S2a--S2c, we also report the mean signed difference between the estimated mapping and the reference-unit conversion $L^{\rm rec}$ over the interior-supported points.
A separate full-grid summary applies the prespecified $X=0$ fallback when the target score falls below the reference score range.

\subsection{Results}
\label{sec:sim-results}

\paragraph{S1a--S1b (correct identification)}
Table~\ref{tab:sim-S1} reports grid-averaged bias, RMSE, and coverage for the absolute prognosis-equivalent mapping under the linear map (S1a) and nonlinear hinge map (S1b), with $m{=}1\to m_{\rm ref}{=}0$.

Under \emph{S1a} (linear), both the absolute and origin-referenced LL mappings are essentially unbiased, and their RMSE decreases at the expected $\sqrt n$ rate.
The analytic delta-method bands attain nominal coverage.
The origin-referenced band is anchored at the origin and widens in proportion to $x$, with all uncertainty carried by the slope $\hat\phi$ (\ref{appB}).
The scalar parameter $\phi$ is also unbiased, with approximately $95\%$ coverage, and the power of the test of $H_0:\phi=1$ increases with $n$, from $0.35$ at $n{=}500$ to $0.79$ at $n{=}2000$.

When treatment assignment is confounded, the absolute LL-OP mapping differs from the reference-unit conversion by about $-0.20$, and its bands cover that conversion only $60\%$ pointwise and $46\%$ simultaneously at $n=2000$.
These are discrepancies from the reference-unit conversion, not estimation bias or coverage for the OP estimand itself.
The corresponding origin-referenced LL-OP mapping removes this level discrepancy and remains close to the reference-unit conversion (Web~Table~\ref{tab:sim-comp-treat}).
In contrast, the absolute PS-stc mapping targets the fixed-treatment quantity and has approximately nominal coverage of its corresponding truth.

Under \emph{S1b} (nonlinear hinge), the linear model is misspecified.
For the absolute LL-OP mapping, pointwise coverage remains close to nominal, but simultaneous coverage decreases with $n$, from $92\%$ to $88\%$ and then to $82\%$, as the irreducible approximation bias of about $+0.05$ becomes detectable.
For the absolute mapping, the spline learner SL maintains simultaneous coverage between $92\%$ and $94\%$, at the cost of higher variance.
With confounded treatment assignment, the absolute SL-PS-stc mapping has simultaneous coverage of $90\%$ for its corresponding truth, compared with $58\%$ coverage of the reference-unit conversion by the absolute SL-OP bands.

\begin{table}[h]
\centering
\caption{S1a--S1b (binary $M$; G1--G4 hold for the without-treatment OP and PS-stc targets): grid-averaged bias, RMSE, and coverage for the absolute prognosis-equivalent mapping ($m{=}1\to m_{\rm ref}{=}0$).
Cov$_{\rm pt}$/Cov$_{\rm sim}$ are pointwise/simultaneous coverage (\%).
For the with-treatment absolute OP mapping, the reported bias, RMSE, and coverage are calculated relative to the reference-unit conversion and are not performance measures for the OP estimand.}
% \caption{S1a--S1b (binary $M$, G1--G4 hold): bias, RMSE, and coverage of the mapping, grid-averaged ($m{=}1\to m_{\rmref}{=}0$).
% Cov$_{\rm pt}$/Cov$_{\rm sim}$ are pointwise/simultaneous coverage (\%).
% For with-treatment OP, the reported bias, RMSE, and coverage are calculated relative to the reference-unit conversion and are not performance measures for the OP estimand.}
\label{tab:sim-S1}
\vspace{2mm}
\begin{tabular}{ll r r r r}
\hline
Estimator & $n$ & Bias & RMSE & Cov$_{\rm pt}$ & Cov$_{\rm sim}$ \\
\hline
\multicolumn{6}{l}{\textit{S1a (linear map), without treatment}}\\
\quad LL-OP & 500  & $-0.007$ & 0.249 & 96 & 96 \\
\quad LL-OP & 1000 & $-0.001$ & 0.164 & 95 & 95 \\
\quad LL-OP & 2000 & $+0.003$ & 0.110 & 96 & 95 \\
\multicolumn{6}{l}{\textit{S1a (linear map), with treatment ($n=2000$)}}\\
\quad LL-OP (confounded) & 2000 & $-0.204$ & 0.240 & 60 & 46 \\
\quad LL-PS-stc          & 2000 & $-0.001$ & 0.140 & 96 & 96 \\
\multicolumn{6}{l}{\textit{S1b (nonlinear hinge), without treatment}}\\
\quad LL-OP & 500  & $+0.064$ & 0.310 & 94 & 92 \\
\quad LL-OP & 1000 & $+0.056$ & 0.219 & 93 & 88 \\
\quad LL-OP & 2000 & $+0.049$ & 0.165 & 90 & 82 \\
\quad SL-OP & 1000 & $-0.082$ & 0.394 & 96 & 94 \\
\quad SL-OP & 2000 & $-0.085$ & 0.307 & 94 & 92 \\
\multicolumn{6}{l}{\textit{S1b (nonlinear hinge), with treatment ($n=2000$)}}\\
\quad LL-OP (confounded) & 2000 & $-0.154$ & 0.240 & 67 & 39 \\
\quad LL-PS-stc          & 2000 & $+0.041$ & 0.197 & 91 & 86 \\
\quad SL-OP (confounded) & 2000 & $-0.296$ & 0.456 & 75 & 58 \\
\quad SL-PS-stc          & 2000 & $-0.101$ & 0.384 & 94 & 90 \\
\hline
\end{tabular}
\end{table}

\paragraph{S2a--S2c (identification asymmetry)}
In all S2 scenarios, the estimator covers its identified estimand at approximately the nominal $95\%$ level (Table~\ref{tab:sim-S2}). 
The main difference across scenarios is how the identified estimand departs from the reference-unit conversion, and this departure is asymmetric between the absolute and origin-referenced mappings. 
Under S2a (G4), only the absolute mapping is offset, by $\gamma_M/\beta_X^{*}\approx-0.71$, whereas the origin-referenced mapping recovers the reference-unit conversion.
Under S2c (G2), the roles are reversed, and the origin-referenced mapping is offset by $c_1/(1+\delta)\approx0.20$. Under S2b (G3), both mappings diverge from the reference-unit conversion, and the gap cannot be removed even with the flexible learner. 
This reflects an identification limit and confirms Proposition~\ref{prop:units}.

Across the S2 scenarios, pointwise coverage on the interior-supported inversion range is approximately nominal (Table~\ref{tab:sim-S2}).
In S2a, the absolute mapping is interior-supported at $x\in\{1.5,2.0,2.5\}$ only, and its supported-point difference is close to the theoretical offset $\gamma_M/\beta_X^{*}\approx-0.71$.
The remaining two points are excluded from regular performance summaries; the full-grid difference includes them using the $X=0$ fallback.
The S2a origin-referenced mapping is interior-supported at all five points and recovers the reference-unit conversion.
Under S2c, the roles are reversed and the origin-referenced mapping has the constant recovery offset $c_1/(1+\delta)\approx0.20$.
Under S2b, both mappings are interior-supported but have an $x$-dependent difference from the reference-unit conversion arising from their slope difference.
These patterns demonstrate target-asymmetric failure of reference-unit recovery.

\begin{table}[h]
\centering
\small
\setlength{\tabcolsep}{2.5pt}
\caption{S2: mean difference from the reference-unit conversion and pointwise coverage ($n=2000$, LL--OP without treatment, grid-averaged). 
Supp gives the number of regular interior-supported inversion points over the five-point evaluation grid. 
The mean signed difference from $L^{\rm rec}$ and Cov$_{\rm pt}$ are averaged only over the interior-supported points; the stated theoretical offsets have the same scope. 
The full-grid mean signed difference instead uses the full requested grid under the prespecified rule that clamps to $X=0$ when the target score is below the reference score range.}
\label{tab:sim-S2}
\vspace{2mm}
\begin{tabular}{llccccc}
\hline
\shortstack{Scenario\\(departure)} & Target & Supp & \multicolumn{2}{c}{\shortstack{Mean difference\\from $L^{\rm rec}$}} & \shortstack{Theoretical\\offset} & \shortstack{Cov$_{\rm pt}$\\(\%)} \\
 & & & \shortstack{Supported\\points} & \shortstack{Full grid\\with fallback} & & \\
\hline
S2a (G4) & abs & 3/5 & $-0.712$ & $-0.619$ & $\gamma_M/\beta_X^{*}=-0.71$ & 94 \\
 & origin & 5/5 & $+0.012$ & $+0.012$ & $0$ & 95 \\
S2b (G3) & abs & 5/5 & $-0.286$ & $-0.286$ & \shortstack{$-0.190x$\\(grid mean $-0.286$)} & 95 \\
 & origin & 5/5 & $-0.280$ & $-0.280$ & \shortstack{$-0.190x$\\(grid mean $-0.286$)} & 95 \\
S2c (G2) & abs & 5/5 & $+0.002$ & $+0.002$ & $0$ & 95 \\
 & origin & 5/5 & $+0.209$ & $+0.209$ & $c_1/(1+\delta)=0.20$ & 96 \\
\hline
\end{tabular}
\end{table}

\paragraph{S3a--S3b (continuous modifier)}
With a continuous modifier (Table~\ref{tab:sim-S3}), S3a is recovered well by LL, which includes the required $X\times M$ interaction. 
The bias is close to zero, and coverage is approximately $95\%$ for every $m_{\rmref}$. RMSE is lowest at the interior reference value $m_{\rmref}=0.5$. 
SL has approximately nominal pointwise coverage, but its simultaneous coverage is lower at $m_{\rmref}=1$ ($83\%$), and it is less efficient than LL. 
In the with-treatment rows, the OP entries are comparisons with the reference-unit conversion.
At the origin reference, the OP mappings differ from that conversion and their bands less often contain it, whereas PS-stc has simultaneous coverage of $94\%$ for LL and $96\%$ for SL relative to its corresponding truth.

Under the hinge scenario S3b, the linear model is again misspecified in $x$. 
At the origin reference, SL has lower bias and better simultaneous coverage than LL, with simultaneous coverage of $95\%$ versus $91\%$ at $n{=}2000$. 
The OP bands likewise less often contain the reference-unit conversion, whereas PS-stc has simultaneous coverage of $93\%$ for LL and $96\%$ for SL relative to its corresponding truth.

\begin{table}[h]
\centering
\caption{S3 (continuous $M$, absolute mapping): bias, RMSE, and coverage, grid-averaged.
S3a (linear, $X\times M$ interaction) is shown across the generalized reference $m_{\rmref}$ at $n=2000$; S3b (hinge in $x$) at the origin reference $m_{\rmref}=0$.
Cov$_{\rm pt}$/Cov$_{\rm sim}$ are pointwise/simultaneous coverage (\%).
For with-treatment OP, the reported bias, RMSE, and coverage are calculated relative to the reference-unit conversion.}
\label{tab:sim-S3}
\vspace{2mm}
\begin{tabular}{ll r r r r}
\hline
Estimator & $m_{\rmref}$ & Bias & RMSE & Cov$_{\rm pt}$ & Cov$_{\rm sim}$ \\
\hline
\multicolumn{6}{l}{\textit{S3a (linear, $X\times M$ interaction), $n=2000$, without treatment}}\\
\quad LL-OP & 0   & $+0.000$ & 0.106 & 96 & 96 \\
\quad SL-OP & 0   & $+0.009$ & 0.288 & 96 & 94 \\
\quad LL-OP & 0.5 & $+0.000$ & 0.047 & 97 & 96 \\
\quad SL-OP & 0.5 & $-0.002$ & 0.100 & 97 & 91 \\
\quad LL-OP & 1   & $+0.039$ & 0.281 & 94 & 94 \\
\quad SL-OP & 1   & $+0.133$ & 0.635 & 94 & 83 \\
\multicolumn{6}{l}{\textit{S3a (linear, $X\times M$ interaction), with treatment ($m_{\rmref}=0$, $n=2000$)}}\\
\quad LL-OP (confounded) & 0 & $-0.107$ & 0.165 & 84 & 79 \\
\quad LL-PS-stc          & 0 & $-0.004$ & 0.140 & 95 & 94 \\
\quad SL-OP (confounded) & 0 & $-0.109$ & 0.326 & 92 & 90 \\
\quad SL-PS-stc          & 0 & $+0.008$ & 0.353 & 96 & 96 \\
\hline
Estimator & $n$ & Bias & RMSE & Cov$_{\rm pt}$ & Cov$_{\rm sim}$ \\
\hline
\multicolumn{6}{l}{\textit{S3b (hinge in $x$), $m_{\rmref}=0$, without treatment}}\\
\quad LL-OP & 1000 & $+0.027$ & 0.165 & 96 & 92 \\
\quad LL-OP & 2000 & $+0.029$ & 0.115 & 94 & 91 \\
\quad SL-OP & 1000 & $+0.000$ & 0.350 & 97 & 96 \\
\quad SL-OP & 2000 & $-0.016$ & 0.234 & 96 & 95 \\
\multicolumn{6}{l}{\textit{S3b (hinge in $x$), $m_{\rmref}=0$, with treatment ($n=2000$)}}\\
\quad LL-OP (confounded) & 2000 & $-0.073$ & 0.150 & 83 & 73 \\
\quad LL-PS-stc          & 2000 & $+0.024$ & 0.145 & 95 & 93 \\
\quad SL-OP (confounded) & 2000 & $-0.128$ & 0.307 & 93 & 87 \\
\quad SL-PS-stc          & 2000 & $-0.023$ & 0.305 & 97 & 96 \\
\hline
\end{tabular}
\end{table}

% ===========================================================================
\section{Discussion}
\label{sec:discussion}
% ===========================================================================

We proposed prognosis-equivalent mapping as a framework for translating a clinical measurement across levels of a patient-level modifier using a time-to-event outcome as the anchor. The prognosis-equivalent mapping is defined by equating a conditional prognostic quantity and inverting the reference-side curve. In contrast to distributional matching or correction of measurement error in regression coefficients, the proposed target directly asks which value on a reference modifier scale carries the same prognosis under an explicitly defined clinical and causal target.

The framework separates two choices that are central in observational clinical data. First, the observed-practice mapping retains the treatment process operating in the data, whereas the policy-standardized mapping evaluates prognosis under a common treatment policy using the g-formula. The former is appropriate when current treatment practice is part of the clinical meaning to be described, while the latter is appropriate when treatment-mediated heterogeneity should be removed. Second, the absolute mapping equates total prognosis, whereas the origin-referenced mapping subtracts the modifier-specific prognostic level at a clinically justified anchor and thereby isolates modification of the prognostic effect of the measurement. The choice among these mappings is therefore part of the estimand definition.

The measurement-scale interpretation requires additional subject-matter assumptions. When the observed measurement is a modifier-specific reading of a latent common-unit quantity and the origin is preserved across modifier levels, the origin-referenced mapping can be interpreted as a reference-unit conversion. Under the stronger assumption of no residual modifier effect after conditioning on the latent quantity and covariates, the absolute mapping has the same interpretation. These assumptions are generally not testable from observed data.

The simulation study supports these distinctions. Under correct specification, the proposed estimators had small bias and approximately nominal coverage. When treatment assignment depended on the measurement and modifier, observed-practice mappings were biased, whereas policy-standardized mappings based on g-computation restored performance. When the prognostic relationship was nonlinear, spline-based inversion captured curvature missed by the linear interaction model, although with increased variability. The simulations also show that departures from a reference-unit conversion can arise from the identified estimand itself. Modifier-specific level effects, failure of origin preservation, and latent $M\times X^*$ interaction affect the absolute and origin-referenced mappings in different ways.

Several practical cautions follow. Policy-standardized mappings require the usual causal assumptions for treatment standardization, including consistency, conditional exchangeability, positivity, and appropriate handling of censoring. All mappings also require a monotone reference-side prognostic curve over the relevant range. Monotone rearrangement can stabilize finite-sample irregularities, but it cannot make a genuinely non-monotone global mapping well defined. Origin-referenced flexible mappings require particular care when the anchor is poorly supported, since the anchor score may be extrapolative. In such cases, unless a clinically meaningful interior anchor is available, an absolute mapping with a flexible prognostic score is often the safer default.

Future work should develop sensitivity analyses for the latent-measurement assumptions, especially imperfect origin preservation and latent score variation. Further extensions should consider policy-standardized mappings for longitudinal treatment regimes and adapt the framework to other prognostic quantities, such as restricted mean survival time, competing-risk cumulative incidence, and absolute risk at multiple time points. These extensions would preserve the main principle of the proposed framework, which is to translate measurements according to the prognosis they imply rather than according to their physical units alone.

\section*{Acknowledgments}
Kazuharu Harada (K.H.) was partially supported by the Japan Society for the Promotion of Science
(JSPS) KAKENHI under Grant Number 25K21165.\vspace*{-8pt}

\bibliographystyle{plainnat}
\bibliography{refs}

% =====================================================================
% Supplementary Material (Web Appendices A--E)
% =====================================================================
\clearpage
\setcounter{section}{0}
\setcounter{figure}{0}
\setcounter{table}{0}
\setcounter{equation}{0}
\renewcommand{\thesection}{Web Appendix \Alph{section}}
\renewcommand{\theequation}{\Alph{section}.\arabic{equation}}
\renewcommand{\thefigure}{\Alph{section}.\arabic{figure}}
\renewcommand{\thetable}{\Alph{section}.\arabic{table}}
\renewcommand{\tablename}{Web Table}
\renewcommand{\figurename}{Web Figure}
\renewcommand{\thepage}{S\arabic{page}}
\setcounter{page}{1}
{\centering\Large\bfseries Supplementary Material\par}
\medskip
{\centering\normalsize for ``Prognosis-equivalent mapping of clinical measurements via survival analysis''\par}
\bigskip

\noindent This Supplementary Material collects material referenced in the main text: proofs and the consistency of the plug-in inverse mappings (\ref{appA}), analytic delta-method inference for the linear observed-practice and fixed-treatment policy-standardized mappings (\ref{appB}), and additional simulation scenarios with the full performance-metric definitions (\ref{appC}).
Equation, assumption, and proposition numbers refer to the main text unless stated otherwise.

% ===========================================================================
\section{Proofs of Propositions 1--3 and consistency of the plug-in mappings}
\label{appA}
% ===========================================================================
\setcounter{equation}{0}

We restate and prove Propositions~1--3 of the main text.
The identification assumptions (A1)--(A4) used by Proposition~2 are stated formally below; the main text gives them in prose.
The maintained prognostic-monotonicity condition (strict monotonicity of the prognostic score, hence of the excess score) makes every supported solution unique, but does not imply existence when the target lies outside the reference score range; it is not a separate assumption.
The measurement special case (Proposition~3) uses the following conditions on the latent model $X=g_M(X^*)$:
\begin{enumerate}
\item[\textnormal{(G1)}] {\bf Invertibility.}
For each modifier value $m$, the measurement map $g_m(x^*)$ is strictly increasing in $x^*$, hence invertible.

\item[\textnormal{(G2)}] {\bf Origin preservation.}
$g_m(0)=0$ for every modifier value $m$, so $X=0$ corresponds to $X^*=0$.

\item[\textnormal{(G3)}] {\bf Latent excess-score invariance.}
There exists a function $\widetilde\Psi_t^{\pi,\circ}(x^*,z)$ that does not depend on $m$ such that
\begin{gather}
\Psi_t^{\pi,\ast,\circ}(x^*,m,z)
=
\widetilde\Psi_t^{\pi,\circ}(x^*,z)
\end{gather}
for every $(x^*,m,z)$ under consideration.

\item[\textnormal{(G4)}] {\bf No residual modifier effect.}
$T^\pi \ci M \mid X^*,Z$. 
\end{enumerate}

\paragraph{Proposition~\ref{prop:recovery} (Equivalence of the absolute and origin-referenced mappings).}
\emph{Under prognostic monotonicity, fix $t$, $(m,z)$, and $m_{\rm ref}$, and suppose that the supported domains of the absolute and origin-referenced mappings have a nonempty intersection.
If $\ell_t(m,z)=\ell_t(m_{\rm ref},z)$, then the two supported domains coincide and
$L_t^{\mathcal Q}(x,m,z;m_{\rm ref})
=
L_t^{\mathcal Q,\circ}(x,m,z;m_{\rm ref})$
at every $x$ in that domain.
Conversely, if the two mappings agree at one point in their common supported domain, then
$\ell_t(m,z)=\ell_t(m_{\rm ref},z)$.
Consequently, the two mappings coincide throughout their supported domains for all $(m,z)$ under consideration if and only if the prognostic level is modifier-invariant.}

\begin{proof}
Decompose $\Psi_t^{\mathcal Q}(x,m,z)=\ell_t(m,z)+\Psi_t^{\mathcal Q,\circ}(x,m,z)$ with $\ell_t(m,z)=\Psi_t^{\mathcal Q}(0,m,z)$ and $\Psi_t^{\mathcal Q,\circ}(0,m,z)=0$.
On their respective supported domains, the origin-referenced mapping $L^\circ$ solves $\Psi_t^{\mathcal Q,\circ}(L^\circ,m_{\rmref},z)=\Psi_t^{\mathcal Q,\circ}(x,m,z)$, while the absolute mapping $L$ solves $\Psi_t^{\mathcal Q}(L,m_{\rmref},z)=\Psi_t^{\mathcal Q}(x,m,z)$, that is, $\ell_t(m_{\rmref},z)+\Psi_t^{\mathcal Q,\circ}(L,m_{\rmref},z)=\ell_t(m,z)+\Psi_t^{\mathcal Q,\circ}(x,m,z)$.
By the maintained monotonicity, $\Psi_t^{\mathcal Q,\circ}(\cdot,m_{\rmref},z)$ has an inverse on its range, so wherever the absolute mapping is supported,
\begin{gather}
    L=\bigl[\Psi_t^{\mathcal Q,\circ}(\cdot,m_{\rmref},z)\bigr]^{-1}
      \!\left\{\Psi_t^{\mathcal Q,\circ}(x,m,z)+\ell_t(m,z)-\ell_t(m_{\rmref},z)\right\}.
\end{gather}
If $\ell_t(m,z)=\ell_t(m_{\rmref},z)$, the absolute and origin target arguments are identical.
Their supported domains therefore coincide, the offset vanishes, and the displayed inverse reduces to $L^\circ$.
Conversely, suppose that the two mappings agree at one $x$ in their common supported domain.
The origin equation gives $\Psi_t^{\mathcal Q,\circ}(L,m_{\rmref},z)=\Psi_t^{\mathcal Q,\circ}(x,m,z)$; subtracting it from the absolute equation gives $\ell_t(m_{\rmref},z)=\ell_t(m,z)$.

Thus level invariance implies equality throughout the coincident supported domains, whereas equality at even one point in a nonempty common supported domain implies level invariance.
When the level is not invariant, the displayed expression carries exactly the offset $\ell_t(m,z)-\ell_t(m_{\rmref},z)$ on the excess scale wherever that inverse exists.
\end{proof}

\paragraph{Identification assumptions (A1)--(A4).}
Let $T^a$ denote the potential event time under the point intervention $A=a$.
Treatment takes values in a discrete set, with no interference and a single version of each treatment, so each $T^a$ is well defined; the policy-counterfactual $T^\pi$ is induced by drawing $A\sim\pi(\cdot\mid x,m,z)$ independently of $\{T^a\}$ given $(x,m,z)$.
These conditions are understood to apply at the source point, at the reference-side points over which the inverse is sought, and, for the origin-referenced mapping, at both anchor points.
The g-formula conditions \citep{Robins1986, HernanRobins2020} are:
\begin{enumerate}
\item[\textnormal{(A1)}] {\bf Consistency.}
If $A=a$, then $T=T^a$.
\item[\textnormal{(A2)}] {\bf Conditional exchangeability.}
$T^a \ci A\mid X,M,Z$ for every treatment value $a$ with $\pi(a\mid x,m,z)>0$ at the target $(x,m,z)$.
\item[\textnormal{(A3)}] {\bf Treatment positivity.}
If $\pi(a\mid x,m,z)>0$, then $P(A=a\mid X=x,M=m,Z=z)>0$ for the target values of $(x,m,z)$.
\item[\textnormal{(A4)}] {\bf Independent censoring and censoring positivity.}
Event times are independently censored conditional on at least $(X,M,Z,A)$, and $P(C\ge t\mid X=x,M=m,Z=z,A=a)>0$ for the target time $t$ and treatment values $a$ with $\pi(a\mid x,m,z)>0$.
\end{enumerate}

\paragraph{Proposition 2 (Identification of $Q_t^\pi$).}
\emph{Under {\rm (A1)--(A4)}, the policy-standardized prognostic quantity satisfies $Q_t^\pi(x,m,z)=\sum_a P(T>t \mid X=x,M=m,Z=z,A=a)\,\pi(a\mid x,m,z)$.
Suppose that these g-formula quantities are identified at the source point and throughout $\mathcal{X}_{m_{\rmref},z}$.
If $Q_t^\pi(\cdot,m_{\rmref},z)$ is strictly monotone in $x$ and $Q_t^\pi(x,m,z)$ belongs to its range over $\mathcal{X}_{m_{\rmref},z}$, the unique policy-standardized mapping $L_t^\pi(x,m,z;m_{\rmref})$ is identified from the observed data.
The analogous conclusion holds for $L_t^{\pi,\circ}$ when both anchor scores are identified and finite and the source excess score belongs to the range of the reference excess score.}

\begin{proof}
By conditional exchangeability (A2), $T^a \ci A \mid (X,M,Z)$, and by consistency (A1) the realized outcome equals $T^a$ on $\{A=a\}$; hence for any $(x,m,z,a)$ with $\pi(a\mid x,m,z)>0$,
\begin{align}
    P(T^a>t \mid X=x,M=m,Z=z)
    &= P(T^a>t \mid X=x,M=m,Z=z,A=a) \notag \\
    &=~ P(T>t \mid X=x,M=m,Z=z,A=a). \label{eq:fixed-a-id}
\end{align}
This is expressed in observed variables only, and the positivity condition (A3) justifies its evaluation at each $(x,m,z,a)$.
The counterfactual $T^\pi$ under the stochastic policy $\pi$ is the potential outcome induced by $A\sim\pi(\cdot\mid x,m,z)$, so
\begin{align}
    Q_t^\pi(x,m,z)
    &=
    P(T^\pi>t \mid X=x,M=m,Z=z) \notag\\
    &=
    \sum_a P(T^a>t \mid X=x,M=m,Z=z)\,\pi(a\mid x,m,z) \notag\\
    &=
    \sum_a P(T>t \mid X=x,M=m,Z=z,A=a)\,\pi(a\mid x,m,z),
\end{align}
where the last line uses \eqref{eq:fixed-a-id}.
Under independent censoring and the censoring-positivity clause of (A4), each $P(T>t\mid X=x,M=m,Z=z,A=a)$ is identified up to $t$ as a functional of the right-censored observed-data law.
Consistent estimation of that functional additionally requires the regularity conditions of the chosen model or learner and is separate from the identification identity established here.
Finally, if the reference curve is strictly monotone and $Q_t^\pi(x,m,z)$ belongs to its range over $\mathcal{X}_{m_{\rmref},z}$, the prognostic-equivalence equation has exactly one solution in that domain.
Because both the source quantity and reference curve are identified there, this solution is identified from the observed data.
For the origin-referenced mapping, additionally require $0<Q_t^\pi(0,m,z)<1$ and $0<Q_t^\pi(0,m_{\rmref},z)<1$, so that both transformed anchor scores are finite.
Subtracting these identified anchor scores gives the same conclusion when the source excess score belongs to the range of the reference excess score; this range condition supplies existence, while strict monotonicity supplies uniqueness.
\end{proof}

\paragraph{Proposition 3 (Reference-unit conversion).}
\emph{Assume the measurement model $X=g_M(X^*)$ and the maintained prognostic monotonicity.
{\rm (i)} Under {\rm (G1)}, {\rm (G2)}, and {\rm (G3)}, the origin-referenced policy-standardized mapping is the reference-unit conversion
$L_t^{\pi,\circ}(x,m,z;m_{\rm ref}) = g_{m_{\rm ref}}\!\left\{g_m^{-1}(x)\right\}$.
{\rm (ii)} Under {\rm (G1)} and {\rm (G4)}, the absolute policy-standardized mapping equals the same conversion; condition {\rm (G2)} is not required.}

\begin{proof}
Conditions (G3) and (G4) are understood to concern the prognostic score induced by the chosen target policy $\pi$.
Recall that the latent-scale prognostic score and its excess relative to the latent origin are
\begin{align}
\Psi_t^{\pi,\ast}(x^*,m,z)
&:=
\Psi_t^\pi\!\left(g_m(x^*),m,z\right), \notag\\
\Psi_t^{\pi,\ast,\circ}(x^*,m,z)
&:=
\Psi_t^{\pi,\ast}(x^*,m,z)
-
\Psi_t^{\pi,\ast}(0,m,z).
\notag
\end{align}
Fix a source value $x$ such that $x^*=g_m^{-1}(x)$ belongs to the common latent domain of $g_m$ and $g_{m_{\rm ref}}$, and suppose that
$g_{m_{\rm ref}}(x^*)\in\mathcal{X}_{m_{\rm ref},z}$.
For part (i), suppose additionally that the source and reference anchors belong to their respective domains and that the corresponding anchor scores are finite.

\textit{(i)}
By (G2), $g_m(0)=0$, and therefore
\begin{align}
\Psi_t^{\pi,\circ}(x,m,z)
&=
\Psi_t^\pi(x,m,z)
-
\Psi_t^\pi(0,m,z) \notag\\
&=
\Psi_t^\pi\!\left(g_m(x^*),m,z\right)
-
\Psi_t^\pi\!\left(g_m(0),m,z\right) \notag\\
&=
\Psi_t^{\pi,\ast}(x^*,m,z)
-
\Psi_t^{\pi,\ast}(0,m,z) \notag\\
&=
\Psi_t^{\pi,\ast,\circ}(x^*,m,z).
\label{eq:proof-latent-source-excess}
\end{align}
Condition (G3) gives
\begin{gather}
\Psi_t^{\pi,\ast,\circ}(x^*,m,z)
=
\widetilde\Psi_t^{\pi,\circ}(x^*,z),
\end{gather}
where $\widetilde\Psi_t^{\pi,\circ}$ does not depend on $m$.
Applying the same argument on the reference side and using $g_{m_{\rm ref}}(0)=0$ gives
\begin{gather}
\Psi_t^{\pi,\circ}
\!\left(g_{m_{\rm ref}}(x^*),m_{\rm ref},z\right)
=
\widetilde\Psi_t^{\pi,\circ}(x^*,z).
\end{gather}
Hence
\begin{gather}
\Psi_t^{\pi,\circ}
\!\left(g_{m_{\rm ref}}(x^*),m_{\rm ref},z\right)
=
\Psi_t^{\pi,\circ}(x,m,z).
\end{gather}
Because the reference-side excess score is strictly monotone on its supported domain, the origin-referenced prognostic-equivalence equation has the unique solution
\begin{gather}
L_t^{\pi,\circ}(x,m,z;m_{\rm ref})
=
g_{m_{\rm ref}}(x^*)
=
g_{m_{\rm ref}}\!\left\{g_m^{-1}(x)\right\}.
\end{gather}

\textit{(ii)}
Under (G4), the counterfactual outcome distribution at fixed $(X^*,Z)$ does not depend on $M$.
Consequently, there exists a function $\widetilde\Psi_t^\pi(x^*,z)$ such that
\begin{gather}
\Psi_t^{\pi,\ast}(x^*,m,z)
=
\widetilde\Psi_t^\pi(x^*,z)
\end{gather}
for every $(x^*,m,z)$ under consideration.
It follows that
\begin{align}
\Psi_t^\pi(x,m,z)
&=
\widetilde\Psi_t^\pi\!\left(g_m^{-1}(x),z\right), \notag\\
\Psi_t^\pi(L,m_{\rm ref},z)
&=
\widetilde\Psi_t^\pi\!\left(g_{m_{\rm ref}}^{-1}(L),z\right).
\notag
\end{align}
The absolute prognostic-equivalence equation therefore becomes
\begin{gather}
\widetilde\Psi_t^\pi\!\left(g_{m_{\rm ref}}^{-1}(L),z\right)
=
\widetilde\Psi_t^\pi\!\left(g_m^{-1}(x),z\right).
\end{gather}
By prognostic monotonicity,
\begin{gather}
g_{m_{\rm ref}}^{-1}(L)
=
g_m^{-1}(x),
\end{gather}
and hence
\begin{gather}
L_t^\pi(x,m,z;m_{\rm ref})
=
g_{m_{\rm ref}}\!\left\{g_m^{-1}(x)\right\}.
\end{gather}
This argument does not reference the origin, so condition (G2) is not required.
\end{proof}

This proof establishes a structural identity on the shared latent domain for the chosen policy; it neither identifies the measurement maps $g_m$ themselves nor guarantees observed-data support.
Estimating the prognosis-equivalent mapping from observed data additionally requires the identification and support conditions of Proposition~2 at the source and converted reference points and, for the origin-referenced mapping, at both anchors.

\paragraph{Consistency of the plug-in inverse mappings.}
\label{app:plugin-consistency}
We justify the consistency statements of the main text (Section~4.3).
The argument is a direct application of uniform convergence, monotone rearrangement, and continuity of the inverse map.
The uniform claims below require a common positive interior margin from the relevant reference-range boundaries and a common modulus of continuity for the associated family of inverse curves.
When the reference curves are differentiable, a uniform lower bound $c>0$ on their absolute slopes over the inversion ranges is a sufficient inverse-conditioning condition.

Consider first the observed-practice mapping.
Throughout this argument, a fitted reference curve is inverted after rearrangement in the maintained monotonicity direction; without rearrangement, the same conclusions require the fitted curve to be monotone with probability tending to one.
Let
\begin{gather}
r_{m_{\rmref},z}(u):=\eta(u,m_{\rmref},z),
\qquad
r_{m,z}(x):=\eta(x,m,z).
\end{gather}
The target mapping satisfies
\begin{gather}
r_{m_{\rmref},z}\{L^\obs(x,m,z;m_{\rmref})\}
=
r_{m,z}(x).
\end{gather}
Suppose that $\hat\eta$ is uniformly consistent on the relevant compact domain and that $r_{m_{\rmref},z}$ is strictly monotone on the relevant compact interval.
Then
\begin{gather}
\sup_{(x,m,z)\in\mathcal{D}}
\left|
\hat r_{m,z}(x)-r_{m,z}(x)
\right|
\overset{p}{\longrightarrow}0,
\end{gather}
and the same holds for the fitted reference curve $\hat r_{m_{\rmref},z}(\cdot)$.
Since directional rearrangement preserves uniform consistency, and since uniform convergence to a continuous strictly monotone limit implies uniform convergence of the corresponding inverses on compact subsets of the interior of the limiting range,
\begin{gather}
\sup_{(x,m,z)\in\mathcal{D}}
\left|
\hat L^\obs(x,m,z;m_{\rmref})
-
L^\obs(x,m,z;m_{\rmref})
\right|
\overset{p}{\longrightarrow}0.
\end{gather}

The origin-referenced estimator is handled by replacing $\eta$ with the centered score
\begin{gather}
\eta^\circ(x,m,z)
:=
\eta(x,m,z)-\eta(0,m,z).
\end{gather}
Uniform consistency of $\hat\eta$ implies uniform consistency of $\hat\eta^\circ$, because the anchor term is evaluated on the same compact domain.
The same inverse-continuity argument therefore gives consistency of $\hat L^{\obs,\circ}$.

For the policy-standardized mapping, assume
\begin{gather}
\sup_{(x,m,z)\in\mathcal{D}}
\left|
\widehat Q_t^\pi(x,m,z)-Q_t^\pi(x,m,z)
\right|
\overset{p}{\longrightarrow}0.
\end{gather}
If the reference curve $Q_t^\pi(\cdot,m_{\rmref},z)$ is strictly monotone on the relevant compact range and satisfies the common inverse-continuity condition above, then the same argument yields
\begin{gather}
\sup_{(x,m,z)\in\mathcal{D}}
\left|
\hat L_t^\pi(x,m,z;m_{\rmref})
-
L_t^\pi(x,m,z;m_{\rmref})
\right|
\overset{p}{\longrightarrow}0.
\end{gather}

For the origin-referenced policy-standardized estimator, define
\begin{gather}
\Psi_t^{\pi,\circ}(x,m,z)
:=
\log\{-\log Q_t^\pi(x,m,z)\}
-
\log\{-\log Q_t^\pi(0,m,z)\}.
\end{gather}
If $Q_t^\pi$ is bounded away from $0$ and $1$ on the relevant range, the map $q\mapsto\log\{-\log q\}$ is Lipschitz on that range.
Hence uniform consistency of $\widehat Q_t^\pi$ implies uniform consistency of $\widehat\Psi_t^{\pi,\circ}$, including at the anchor $x=0$.
Consistency of the origin-referenced inverse then follows from the same monotone inverse argument.

Rearrangement in the maintained monotonicity direction is non-expansive in sup-norm, so the rearranged reference curve has the same limiting inverse on the interior supported range.
These conclusions concern inversion of the continuum curve, or of an interpolation that is exact in the limit.
If inversion is performed only on a numerical grid with mesh width $h_n$, a fixed grid targets its discretized or interpolated inverse, whereas consistency for the continuum mapping additionally requires the interpolation error to vanish, for example through $h_n\to0$.

% ===========================================================================
\section{Analytic inference for linear single-score mappings}
\label{appB}
% ===========================================================================
\setcounter{equation}{0}

The derivation in this appendix applies both to the observed-practice mapping and to a policy-standardized mapping under a fixed treatment regime $\pi=\delta_{a_0}$.
For the fixed regime, the g-formula reduces to the treatment-specific Cox survival at $A=a_0$, so prognostic equivalence again reduces to equality of a single fitted Cox score.
Suppose that the relevant score, after collecting terms that remain when $A=a_0$ is fixed, has the linear interaction form
\begin{gather}
    \eta(x,m,z) = \beta_X x + \beta_M m + \beta_{XM}\,m x + \gamma^\top z,
    \label{eq:linear-model-inference}
\end{gather}
where terms depending only on the fixed arm and the common covariate value cancel from the mapping equation.
For OP, $(\beta_X,\beta_M,\beta_{XM})$ are the corresponding coefficients of the OP score.
For fixed-treatment PS, let $\theta_Y$ denote the coefficient vector of the treatment-inclusive outcome model and write the effective coefficient vector as $b(a_0)=R_{a_0}\theta_Y$, where the known matrix $R_{a_0}$ selects and combines any treatment-interaction terms at $a_0$.
Then $\widehat{\mathrm{Var}}\{\hat b(a_0)\}=R_{a_0}\hat V_YR_{a_0}^\top$, and all calculations below apply by the chain rule with the effective coefficients substituted for $(\beta_X,\beta_M,\beta_{XM})$.

All pair-specific derivations below are written for a scalar modifier and concern a fixed nontrivial pair $m\ne m_{\rmref}$; they assume that the source and reference slopes have a common nonzero sign and that the reference slope is bounded away from zero.
We also assume root-$n$ asymptotic normality of the relevant Cox coefficient estimator and consistency of its covariance estimator under the identification and censoring conditions appropriate to the OP or fixed-treatment PS target.
In the formulas below, $\hat V(\hat\theta)$ denotes a consistent estimator of $\mathrm{Var}(\hat\theta)$.
The inverse observed information is appropriate when the Cox model is correctly specified; under misspecification, inference for the pseudo-true working-score mapping instead requires a suitable sandwich covariance estimator.
Let $\alpha\in(0,1)$ denote the nominal noncoverage probability, and define $z_p:=\Phi^{-1}(p)$.

\paragraph{Notation and estimators.}
For brevity, let
\begin{gather}
    c_1 = \beta_M(m - m_{\rmref}),
    \quad
    c_2 = \beta_X + \beta_{XM} m,
    \quad
    c_3 = \beta_X + \beta_{XM} m_{\rmref}.
\end{gather}
Define the population mapping slope and intercept by
\begin{gather}
    \phi
    := \frac{c_2}{c_3},
    \qquad
    \alpha
    := \frac{c_1}{c_3},
\end{gather}
so that the population origin-referenced and absolute mappings are, respectively, $L^\circ(x)=\phi x$ and $L(x)=\alpha+\phi x$.
Here $\phi$ denotes the mapping slope and is distinct from the prognostic transformation $\varphi$, while $\alpha$ denotes the mapping intercept only.
Their plug-in estimators are
\begin{gather}
    \hat\phi
    := \frac{\hat{c}_2}{\hat{c}_3}
    = \frac{\hat\beta_X + \hat\beta_{XM} m}{\hat\beta_X + \hat\beta_{XM} m_{\rmref}},
    \qquad
    \hat\alpha
    := \frac{\hat{c}_1}{\hat{c}_3}
    = \frac{\hat\beta_M(m - m_{\rmref})}{\hat\beta_X + \hat\beta_{XM} m_{\rmref}},
    \label{eq:alpha-phi-def}
\end{gather}
and hence $\hat L^\circ(x)=\hat\phi x$ and $\hat L(x)=\hat\alpha+\hat\phi x$.
For a scalar function $g(\theta)$, the delta-method variance estimator is $\widehat{\mathrm{Var}}\{g(\hat\theta)\}=(\nabla_\theta g\vert_{\hat\theta})^\top\hat V(\hat\theta)(\nabla_\theta g\vert_{\hat\theta})$.

\paragraph{Inference for the origin-referenced mapping.}
The gradient of the functional $\phi=c_2/c_3$, evaluated at $\hat\theta$, is
\begin{gather}
    \nabla_\theta \phi\big|_{\hat\theta}
    =
    \frac{1}{\hat{c}_3} \nabla_\theta c_2\big|_{\hat\theta}
    -
    \frac{\hat{c}_2}{\hat{c}_3^2} \nabla_\theta c_3\big|_{\hat\theta},
    \label{eq:grad-phi}
\end{gather}
and its standard error is
\begin{gather}
    \widehat{\mathrm{se}}(\hat\phi)
    =
    \left\{
        (\nabla_\theta \phi)^\top \hat{V}(\hat\theta)\, (\nabla_\theta \phi)
    \right\}^{1/2}.
    \label{eq:se-phi}
\end{gather}

\textit{Confidence interval for the mapping slope $\phi$:} the Wald $1-\alpha$ interval is
\begin{gather}
    \hat\phi \pm z_{1-\alpha/2}\, \widehat{\mathrm{se}}(\hat\phi).
    \label{eq:ci-phi}
\end{gather}

\textit{Pointwise interval at a fixed point $x_0$:} since $\widehat{\mathrm{se}}\{\hat L^\circ(x_0)\}=|x_0|\widehat{\mathrm{se}}(\hat\phi)$, the pointwise interval is $\hat\phi x_0\pm z_{1-\alpha/2}|x_0|\widehat{\mathrm{se}}(\hat\phi)$.
At $x_0=0$, the origin-referenced mapping, its estimator, and its estimated variance are all fixed at zero by construction, so this degenerate anchor carries no pointwise inferential content and is excluded from coverage summaries.

\textit{Simultaneous band for the mapping curve:} the uncertainty of $L^\circ(x)=\phi x$ is carried by the scalar $\phi$ alone.
On a fixed interval $[x_l,x_u]\subset(0,\infty)$, the band
\begin{gather}
    \hat\phi\, x \pm z_{1-\alpha/2}\, x\, \widehat{\mathrm{se}}(\hat\phi),
    \quad x \in [x_l, x_u]
    \label{eq:band-origin-linear}
\end{gather}
is obtained by scaling the Wald interval \eqref{eq:ci-phi} by $x$ and therefore has asymptotic simultaneous coverage $1-\alpha$ over that interval under the stated regularity conditions.
This is the origin-referenced specialization of the Gaussian plug-in sup-$t$ band introduced below: because all standardized errors are generated by the same scalar $\hat\phi-\phi$, the limiting sup-$t$ critical value reduces to $z_{1-\alpha/2}$.

\paragraph{Inference for the absolute mapping.}
The uncertainty of $\hat L(x)=\hat\alpha+\hat\phi x$ is distributed over the two-dimensional estimator $(\hat\alpha,\hat\phi)$.
The gradient of the functional $\alpha=c_1/c_3$, evaluated at $\hat\theta$, is
\begin{gather}
    \nabla_\theta \alpha\big|_{\hat\theta}
    =
    \frac{1}{\hat{c}_3} \nabla_\theta c_1\big|_{\hat\theta}
    -
    \frac{\hat{c}_1}{\hat{c}_3^2} \nabla_\theta c_3\big|_{\hat\theta},
\end{gather}
and by the delta method the elements of the estimated covariance matrix $\hat{\bs{\Sigma}} \in \doubleR^{2 \times 2}$ of $(\hat\alpha,\hat\phi)$ are
\begin{align}
    \hat\sigma_\alpha^2
    &= (\nabla_\theta \alpha)^\top \hat{V}(\hat\theta)\, (\nabla_\theta \alpha),
    \quad
    \hat\sigma_{\alpha\phi}
    = (\nabla_\theta \alpha)^\top \hat{V}(\hat\theta)\, (\nabla_\theta \phi),
    \quad
    \hat\sigma_\phi^2
    = (\nabla_\theta \phi)^\top \hat{V}(\hat\theta)\, (\nabla_\theta \phi).
    \label{eq:sigma-elements}
\end{align}

\textit{Confidence interval for the mapping intercept $\alpha$:} writing $\hat\sigma_\alpha=(\hat\sigma_\alpha^2)^{1/2}$, the scalar Wald $1-\alpha$ interval is $\hat\alpha\pm z_{1-\alpha/2}\hat\sigma_\alpha$.

\textit{Pointwise interval at a fixed point $x_0$:} the variance estimator of $\hat L(x_0)=\hat\alpha+\hat\phi x_0$ is
\begin{gather}
    \widehat{\mathrm{Var}}\{\hat L(x_0)\}
    =
    \hat\sigma_\alpha^2
    + 2 x_0\, \hat\sigma_{\alpha\phi}
    + x_0^2\, \hat\sigma_\phi^2,
    \label{eq:var-full}
\end{gather}
and the pointwise Wald interval is $\hat L(x_0)\pm z_{1-\alpha/2}\{\widehat{\mathrm{Var}}\{\hat L(x_0)\}\}^{1/2}$.

\textit{Simultaneous band for the mapping curve (Gaussian plug-in sup-$t$):} because $L(x)=\alpha+\phi x$ depends on two parameters, its standardized grid errors need not be perfectly correlated and the simultaneous critical value is calibrated \citep{MontielOleaPlagborgMoller2019}.
On a fixed grid $x_1<\cdots<x_K$, let $\hat{\bs{\Sigma}}_L=J\hat{\bs{\Sigma}}J^\top$ be the implied covariance of $\bigl(\hat L(x_1),\dots,\hat L(x_K)\bigr)$, where $J$ is the $K\times2$ Jacobian with $k$th row $(1,x_k)$, and let $\hat R$ be the corresponding correlation matrix.
Conditionally drawing $\xi\sim N(0,\hat R)$, let $c_{1-\alpha}$ be the $1-\alpha$ quantile of $\max_{1\le k\le K}|\xi_k|$.
The simultaneous band is
\begin{gather}
    \hat L(x_k)
    \pm
    c_{1-\alpha}\,
    \bigl\{\widehat{\mathrm{Var}}\{\hat L(x_k)\}\bigr\}^{1/2},
    \qquad k = 1,\dots,K.
    \label{eq:band-full}
\end{gather}
Under the stated regularity conditions and consistent covariance estimation, this plug-in band has asymptotic simultaneous coverage $1-\alpha$ on the fixed grid.
Its Gaussian critical value satisfies $z_{1-\alpha/2}\le c_{1-\alpha}\le\sqrt{\chi^2_{2,1-\alpha}}$, with the upper bound corresponding to the Scheff\'e construction based on the full two-parameter ellipse \citep{Scheffe1953}.
The band $\hat L(x)\pm\sqrt{\chi^2_{2,1-\alpha}\widehat{\mathrm{Var}}\{\hat L(x)\}}$ is therefore an asymptotically valid, generally conservative closed-form alternative for the fixed grid.

\paragraph{Band width of the origin-referenced and absolute mappings.}
The origin-referenced pointwise variance is $x^2\hat\sigma_\phi^2$, so band \eqref{eq:band-origin-linear} widens from the origin in proportion to $|x|$.
The absolute pointwise variance is $\hat\sigma_\alpha^2+2x\hat\sigma_{\alpha\phi}+x^2\hat\sigma_\phi^2$ and therefore also reflects intercept uncertainty and its covariance with the slope.
Because the cross term can be negative, there is no universal ordering of the two band widths, although the absolute mapping retains nonzero intercept uncertainty at $x=0$ whenever $\hat\sigma_\alpha^2>0$.
The origin-referenced simultaneous band uses $z_{1-\alpha/2}$, whereas the absolute band uses the calibrated $c_{1-\alpha}\in[z_{1-\alpha/2},\sqrt{\chi^2_{2,1-\alpha}}]$.

\paragraph{Pair-specific identity tests.}
For the fixed pair $(m,m_{\rmref})$, identity of the origin-referenced mapping is the scalar null $H_0^\circ:\phi=1$, whereas identity of the absolute mapping is the joint null $H_0:\alpha=0,\ \phi=1$.

\textit{Origin-referenced pair:} the Wald statistic
\begin{gather}
    W_\phi
    =
    \frac{(\hat\phi - 1)^2}{\widehat{\mathrm{Var}}(\hat\phi)}
    \xrightarrow{d} \chi^2_1
    \quad \text{under } H_0^\circ,
    \label{eq:wald-phi}
\end{gather}
tests the one-dimensional pair-specific contrast $c_2=c_3$.

\textit{Absolute pair:} the Wald statistic
\begin{gather}
    W_{(\alpha,\phi)}
    =
    \bigl(\hat\alpha,\; \hat\phi - 1\bigr)\,
    \hat{\bs{\Sigma}}^{-1}
    \begin{pmatrix}\hat\alpha \\ \hat\phi - 1\end{pmatrix}
    \xrightarrow{d} \chi^2_2
    \quad \text{under } H_0,
    \label{eq:wald-full}
\end{gather}
tests the two-dimensional pair-specific contrasts $c_1=0$ and $c_2=c_3$.
These Wald limits require the displayed contrast covariance matrices to be nonsingular.

\paragraph{Global tests for a categorical modifier.}
If $M$ has $K$ levels and is represented by $K-1$ indicators, let $\beta_{XM}\in\doubleR^{K-1}$ and $\beta_M\in\doubleR^{K-1}$ denote the corresponding interaction and main-effect vectors.
The global origin-referenced null is $H_{0,G}^\circ:\beta_{XM}=0$, and the global absolute-identity null is $H_{0,G}:\beta_{XM}=0,\ \beta_M=0$.
Under correct model specification, these restrictions may be tested by Wald or likelihood-ratio tests with, respectively, $K-1$ and ordinarily $2(K-1)$ degrees of freedom.
For a misspecified working model, a sandwich Wald test for the corresponding pseudo-true restrictions is required rather than the usual likelihood-ratio calibration.
These global nulls are distinct from a test for one selected pair and coincide with it only in the corresponding binary single-contrast case.
When several pair-specific hypotheses are reported, multiplicity should be handled separately from the single global test.

% ===========================================================================
\section{Additional simulation scenarios and performance metrics}
\label{appC}
% ===========================================================================
\setcounter{equation}{0}

This appendix collects the full data-generating structure for all reported simulation scenarios, the inference and computational scale of the study, the performance-metric definitions, and the comprehensive per-cell results.
The learners (LL, SL), policy variants (OP, PS-stc), and treatment axis are summarized in the main text.

\paragraph{Common data-generating setup.}
The generating structure common to all scenarios is as follows (common parameters $\beta_X^{*}=0.7$, $\beta_Z=0.5$, $\delta=0.5$, $\tau=-0.5$).
Here $\beta_X^{*}$ is the data-generating prognostic slope on the latent common-unit quantity $X^*$, distinct from the observed-scale coefficient $\beta_X$ that the main-text working model fits; the two differ through the measurement map $X=g_M(X^*)$.
The modifier coefficients $\gamma_M$ (level) and $\gamma_{MX}$ ($M\times X^*$) introduced below are likewise data-generating parameters.
\begin{itemize}
    \item \textbf{Modifier}: in the binary scenarios (S1a--S1b, S2a--S2c) $M \sim \mathrm{Bernoulli}(0.3)$; in the continuous scenarios (S3a, S3b) $M \sim \mathrm{Uniform}(0,1)$.
    \item \textbf{Covariate}: $Z \sim \mathcal{N}(0,1)$.
    \item \textbf{Measurement construction}: in the S1a \emph{linear} map and the S2 scenarios, the latent $X^* \sim \mathrm{LogNormal}(0,0.4)$ is drawn and the observed $X=g_M(X^*)$ follows the scenario map.
    In the matched-marginal scenarios---the nonlinear smooth-hinge S1b, the continuous bilinear S3a, and the continuous smooth-hinge S3b---the observed $X \sim \mathrm{LogNormal}(0,0.4)$ is instead drawn from a common marginal for every $m$ and the latent driver is set $X^*=g_m^{-1}(X)$, so the $M$-conditional distributions of $X$ coincide and a contrast is not a support artefact; the hazard depends on $X^*$ only.
    \item \textbf{Baseline hazard}: Weibull, $\Lambda_0(t)=(t/\sigma)^{\kappa}$ ($\kappa=1.5$, $\sigma=3.0$); event times $T = \sigma\{-\log U/\exp\eta\}^{1/\kappa}$, $U\sim\mathrm{Uniform}(0,1)$.
    \item \textbf{Linear predictor}: $\eta$ takes $X^*,M,X^*\times M,Z$ (and $A$ in the with-treatment version) as inputs, with the covariate entering additively as $\beta_Z Z$.
    This common additive term cancels when a single score is equated, as in the no-treatment and fixed-treatment mappings, but need not cancel from the treatment mixture underlying with-treatment OP; the reported comparison is evaluated at $Z=0$.
    \item \textbf{Treatment} (with-treatment version only): $A \sim \mathrm{Bernoulli}(\mathrm{logit}^{-1}(a_0+a_X X+a_M M+a_Z Z))$ with $a_0=-1$, $a_X=0.3$, $a_M=1.0$, $a_Z=0.3$; the hazard includes $\tau A$ ($\tau=-0.5$).
    \item \textbf{Censoring}: $C \sim \mathrm{Uniform}(0, c_{\max})$ with $c_{\max}=10$; $Y=\min(T,C)$, $\Delta=\mathbf{1}\{T\le C\}$.
    \item \textbf{Evaluation grid}: $x \in \{0.5,1.0,1.5,2.0,2.5\}$; for binary $M$, $m=1\to m_{\rmref}=0$; for continuous $M$, the reported OP and PS-stc summaries use $m\in\{0.25,0.5,0.75\}$ by $m_{\rmref}\in\{0,0.5,1\}$ in S3a and $m_{\rmref}=0$ in S3b.
    All mappings are evaluated at $Z=0$.
\end{itemize}

\paragraph{Scenario-specific generating structures.}
The measurement maps and the resulting true mappings are illustrated in Web Figures~\ref{webfig:scenarios-binary} (binary $M$) and~\ref{webfig:scenarios-continuous} (continuous $M$).
\begin{itemize}
\item \emph{S1a (linear):} measurement map $X=(1+\delta M)X^*$ and a hazard with no direct effect of $M$, $\lambda(t\mid X^*,Z)=\lambda_0(t)\exp(\beta_X^{*}X^*+\beta_Z Z)$; (G1)--(G4) hold for the no-treatment OP and PS-stc targets.
\item \emph{S2a (level shift; G4 violated):} the S1a map with a direct level effect, $\lambda=\lambda_0\exp(\beta_X^{*}X^*+\gamma_M M+\beta_Z Z)$, $\gamma_M=-0.5$.
\item \emph{S2b ($M\times X^*$ interaction; G3 violated):} $\lambda=\lambda_0\exp[(\beta_X^{*}+\gamma_{MX}M)X^*+\beta_Z Z]$, $\gamma_{MX}=-0.20$.
\item \emph{S2c (origin non-preservation; G2 violated):} an intercept under $M=1$, $X=c_1 M+(1+\delta M)X^*$, $c_1=0.30$ (hazard as in S1a).
\item \emph{S1b (nonlinear smooth hinge):} construction as above with latent driver $X^*=h(X)$ for the source arm and $X^*=X$ for the reference arm (hazard as in S1a), where $h$ is the smooth concave hinge
        \begin{gather}
          h(x)=b(x)+\{a_y-b(a_x)\},\qquad b(u)=s_0 u-(s_0-s_1)\,w\log\!\bigl(1+e^{(u-c_1)/w}\bigr),
          \label{eq:hinge}
        \end{gather}
with slopes $s_0=1$, $s_1=1/10$, knot $c_1=1.3$, corner width $w=0.45$, anchored through the origin $(a_x,a_y)=(0,0)$.
Then $h$ is strictly increasing with $h(0)=0$, and $h'(x)<1$ gives $h(x)<x$ for $x>0$; for the no-treatment OP and PS-stc targets, (G1)--(G4) hold and both mappings coincide, $L_{\mathrm{true}}(x;1\to0)=h(x)$.
\item \emph{S3a (continuous, linear in $x$ and the $X\times M$ interaction, no $M$ main effect):} the construction with $X^*=\kappa(m)X$, $\kappa(m)=1-m/3$, so $\eta(x,m,z)=\beta_X^{*}\kappa(m)x+\beta_Z z$ has $(x,m)$-part $\beta_X^{*}x-\tfrac{\beta_X^{*}}{3}\,mx$, linear in $x$ and in the $X\times M$ interaction with no main effect of $M$; for the no-treatment OP and PS-stc targets, (G4) holds and $L_{\mathrm{true}}(x;m\to m_{\rmref})=x\,\kappa(m)/\kappa(m_{\rmref})$.
\item \emph{S3b (continuous, hinge in $x$):} the construction with $X^*=(1-m)X+m\,h(X)$, so $\eta(x,m,z)=\beta_X^{*}\{(1-m)x+m\,h(x)\}+\beta_Z z$ is linear in $m$ and a hinge in $x$; the source map at $m=1$ equals the binary hinge $h$.
For the no-treatment OP and PS-stc targets, the true mapping is $L_{\mathrm{true}}(x;m\to 0)=(1-m)x+m\,h(x)$ at the origin reference and solves $g_{m_{\rmref}}^{-1}(L)=g_m^{-1}(x)$ numerically otherwise; standardizing to a non-origin, compressed reference inverts the near-flat $m=1$ arm (the true $L$ reaching about $13$ at the grid edge), which is ill-conditioned, so we report $m_{\rmref}=0$ only.
\end{itemize}
A spline-origin caveat applies across the spline-fitted scenarios (S1b, S3a, S3b): because origin-referencing a fitted $X$-spline at the anchor $x=0$ extrapolates beyond the observed support, the absolute mapping is the headline target there, while the origin-referenced mapping is still reported in the comprehensive tables to make that instability explicit (its bias and under-coverage for the SL learner).

\begin{figure}[htbp]
\centering
\includegraphics[width=0.8\textwidth]{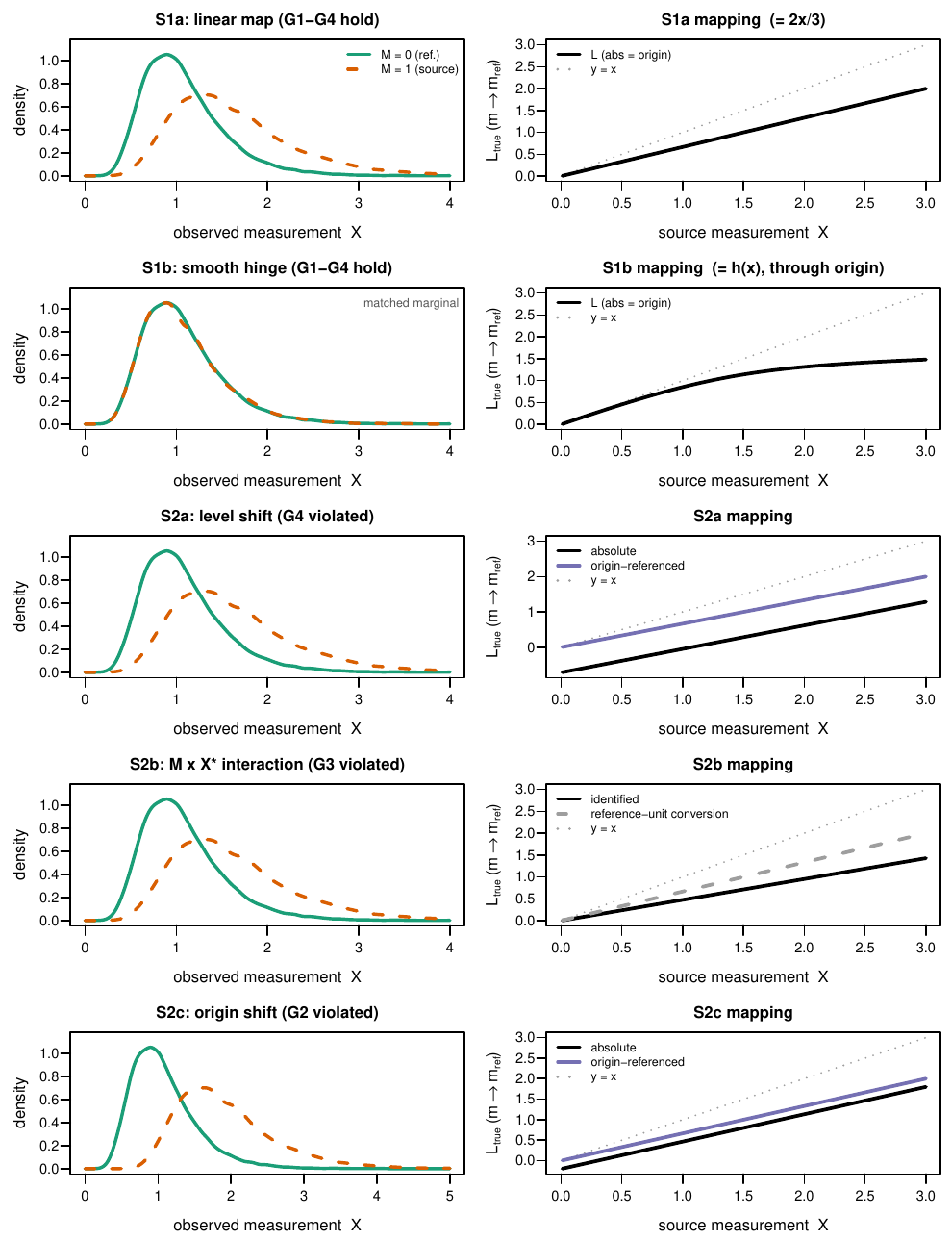}
\caption{Binary-modifier simulation scenarios ($m=1\to m_{\rmref}=0$).
For each scenario the left panel shows the distribution of the observed measurement $X$ by modifier value and the right panel the true prognosis-equivalent mapping $L_{\mathrm{true}}$ against the identity $y=x$.
For the no-treatment OP and PS-stc targets, the (G1)--(G4) maps S1a (linear) and S1b (smooth hinge) both pass through the origin, so the absolute and origin-referenced mappings coincide; S1b draws $X$ from a common marginal for both arms (matched marginals).
The targeted-departure maps S2a (G4), S2b (G3), and S2c (G2) show the asymmetric departure of the absolute and origin-referenced mappings from the reference-unit conversion (Proposition~3).}
\label{webfig:scenarios-binary}
\end{figure}

\begin{figure}[htbp]
\centering
\includegraphics[width=0.8\textwidth]{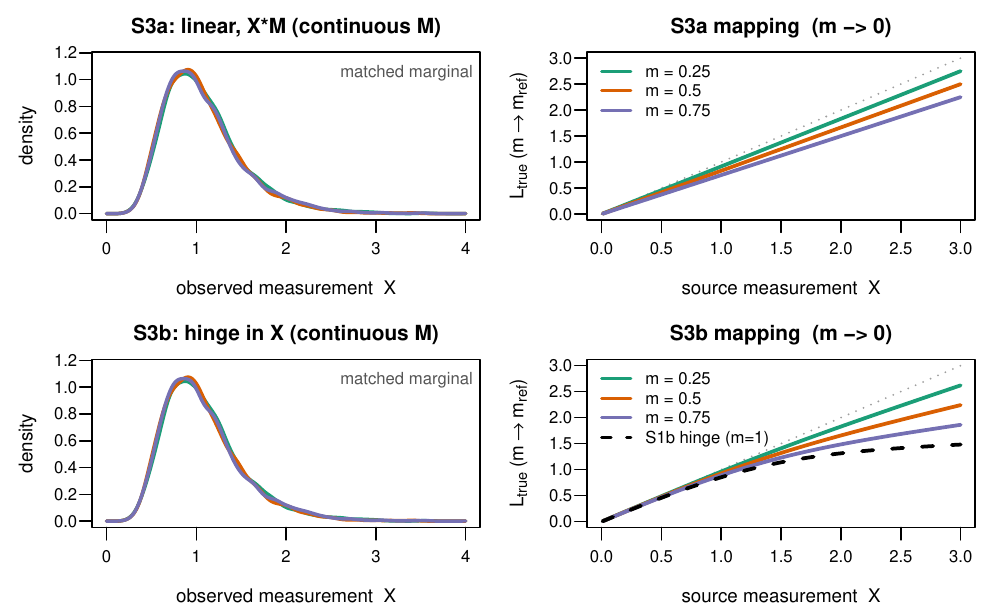}
\caption{Continuous-modifier simulation scenarios ($m\in\{0.25,0.5,0.75\}\to m_{\rmref}=0$), both generated with matched observed-$X$ marginals.
S3a is linear in $x$ and the $X\times M$ interaction with no $M$ main effect; S3b is linear in $m$ and a hinge in $x$, its $m=1$ source map equal to the S1b binary hinge (dashed).
All maps pass through the origin.}
\label{webfig:scenarios-continuous}
\end{figure}

\paragraph{Inference and computational scale.}
The mode of inference and the scale of the study differ by estimator (Web Table~\ref{tab:sim-inference}).
LL under OP and PS-stc admits the inexpensive analytic delta-method bands of \ref{appB}---pointwise Wald intervals and the Gaussian plug-in sup-$t$ simultaneous band---so these cells are evaluated at $R=1000$ replications and $n\in\{500,1000,2000\}$ without resampling.
The spline learner SL under OP and PS-stc relies on the nonparametric bootstrap---pointwise percentile intervals and the calibrated percentile simultaneous band (common tail probability $\zeta$) of Section~5 of the main text---so we use $R=500$ replications at $n\in\{1000,2000\}$ with $B=200$ bootstrap replicates.
For spline learners, the centering and natural-spline basis, including its knots and boundary knots, are constructed from the outer sample and held fixed while individuals are resampled.
Within this conditional-on-basis bootstrap, the Cox model, rearrangement, and inversion are recomputed in every draw.
The recorded failure rate---defined as a missing estimator output or a nonfinite absolute-map value---was zero in every scenario and cell.

\begin{table}[htbp]
\centering
\caption{Mode of inference for each simulation estimator (learner $\times$ policy).
A cell uses the nonparametric bootstrap if and only if its learner is spline-based (SL); the linear learner LL under the observed-practice (OP) or static-policy (PS-stc) mapping uses the analytic delta method instead.
Pointwise intervals are $95\%$, and the simultaneous-band algorithm uses $95\%$ as its target bootstrap content.
``Wald'' is the delta-method normal interval and ``Gaussian plug-in sup-$t$'' the analytic simultaneous band of \ref{appB}; ``percentile'' is the bootstrap percentile interval and ``calibrated percentile'' the simultaneous band of Section \ref{sec:inference} of the main text---the equal-tailed percentile rectangle at the calibrated common tail probability $\zeta$ \citep{MontielOleaPlagborgMoller2019}.
$R$ is the number of Monte Carlo replications and $B$ the number of bootstrap resamples.
For LL, the scalar parameters $\phi,\alpha$ and the test of $H_0:\phi=1$ likewise use analytic delta-method (Wald) inference.}
\label{tab:sim-inference}
\begin{tabular}{ll l l c}
\hline
Estimator & Pointwise CI & Simultaneous band & $R$ / $n$ & $B$ \\
\hline
\multicolumn{5}{l}{\textit{Analytic (delta method; no resampling)}}\\
\quad LL-OP     & Wald       & Gaussian plug-in sup-$t$ & $1000$ / $\{500,1000,2000\}$ & --- \\
\quad LL-PS-stc & Wald       & Gaussian plug-in sup-$t$ & $1000$ / $\{500,1000,2000\}$ & --- \\
\multicolumn{5}{l}{\textit{Nonparametric bootstrap}}\\
% \quad LL-PS-dyn & percentile & calibrated percentile    & $500$ / $\{1000,2000\}$       & $200$ \\
\quad SL-OP     & percentile & calibrated percentile    & $500$ / $\{1000,2000\}$       & $200$ \\
\quad SL-PS-stc & percentile & calibrated percentile    & $500$ / $\{1000,2000\}$       & $200$ \\
% \quad SL-PS-dyn & percentile & calibrated percentile    & $500$ / $\{1000,2000\}$       & $200$ \\
\hline
\end{tabular}
\end{table}

\paragraph{Performance metrics.}
For each evaluation point $k$ (rows indexed by $x$, or by the grid $(x,m,m_{\rmref})$ for continuous $M$), we aggregate the following over $R$ replications (with estimate $\hat L_k^{(r)}$, truth $L_k$, reported SE $\widehat{\rm se}_k^{(r)}$, and 95\% CI $[\ell_k^{(r)},u_k^{(r)}]$): $\mathrm{bias}=\tfrac1R\sum_r(\hat L_k^{(r)}-L_k)$, $\mathrm{RMSE}=\{\tfrac1R\sum_r(\hat L_k^{(r)}-L_k)^2\}^{1/2}$, empirical SE (empSE) $=\mathrm{sd}_r(\hat L_k^{(r)})$, mean reported SE (meanSE) $=\tfrac1R\sum_r\widehat{\rm se}_k^{(r)}$, pointwise coverage $=\tfrac1R\sum_r\mathbf1\{\ell_k^{(r)}\le L_k\le u_k^{(r)}\}$, and width $=\tfrac1R\sum_r(u_k^{(r)}-\ell_k^{(r)})$.
Let $\mathcal K_{\rm sup}$ denote the set of interior-supported evaluation points, whose target scores lie in the interior of the reference score range, and let $[\ell_{k,\mathrm{sim}}^{(r)},u_{k,\mathrm{sim}}^{(r)}]$ denote the corresponding simultaneous-band endpoints.
The simultaneous coverage is $\tfrac1R\sum_r\mathbf1\{\forall k\in\mathcal K_{\rm sup}:\ \ell_{k,\mathrm{sim}}^{(r)}\le L_k\le u_{k,\mathrm{sim}}^{(r)}\}$.
For LL we additionally aggregate the bias, RMSE, empSE, and coverage of the scalars $\phi$ and $\alpha$ and the rejection rate of $H_0:\phi=1$ (Type I error when $\phi_{\rm true}=1$, power when $\phi_{\rm true}\ne1$).
For continuous $M$ (S3a, S3b) the grid-averaged ($x\times$ source-$m$) metrics are reported by $m_{\rmref}$.

These primary performance metrics are defined only for points at which the target score belongs to the interior of the reference score range; the simultaneous event likewise ranges only over that interior-supported subset.
An attainable boundary root defines a mapping but is excluded from these regular performance summaries.
If the target lies below the reference range, the simulation's practical output applies the prespecified lower-bound fallback $X=0$, projecting the point estimate and both pointwise and simultaneous interval endpoints to the nonnegative measurement domain.
``Supp'' reports the number of interior-supported points over the requested grid size, and ``T-clamp'' the fraction of grid truths requiring the fallback.
For analytic LL--OP and LL--PS-stc, the reported SE remains the unprojected working-model SE.
The full-grid difference from the reference-unit conversion is kept separate from interior-supported bias and coverage; no regular coverage interpretation is assigned to a boundary-clamped interval.

\paragraph{Comprehensive results.}
Web Tables~\ref{tab:sim-comp-obs}--\ref{tab:sim-comp-params} report every scenario $\times$ learner $\times$ reported policy $\times$ target $\times$ $n$ ($\times\,m_{\rmref}$) cell: the without-treatment mapping accuracy (Web Table~\ref{tab:sim-comp-obs}), the difference from the reference-unit conversion for every reported method (Web Table~\ref{tab:sim-comp-recovery}), the with-treatment OP versus fixed-treatment PS accuracy (Web Table~\ref{tab:sim-comp-treat}), and the scalar $\phi/\alpha$ calibration for LL in the linear scenarios S1a and S2 (Web Table~\ref{tab:sim-comp-params}).
Both the absolute and origin-referenced targets are reported.
For the spline learner SL the origin-referenced mapping is obtained by anchoring the fitted $X$-spline at $x=0$, an extrapolation beyond the observed support; SL-origin is reported nonetheless to make this instability explicit---its bias and under-coverage relative to the well-behaved SL-absolute and LL-origin mappings---rather than to silently suppress it.
They corroborate the main-text summaries: the linear $X\times M$-interaction continuous scenario S3a is recovered by the interaction-capable LL at every $m_{\rmref}$ (most accurately at the interior $m_{\rmref}=0.5$); under the hinge scenarios the spline learner SL holds simultaneous coverage under the nonlinear inversion at the cost of higher variance (reducing the inversion bias in the continuous S3b), while the linear model's simultaneous coverage degrades with $n$ as the approximation bias becomes detectable; and policy-standardized estimators are evaluated against their corresponding common-policy truths.
For the no-treatment OP mapping and the PS-stc mapping, the difference from the reference-unit conversion is negligible under G1--G4 (S1, S3), confirming that the prognosis-equivalent mapping coincides with that conversion, whereas the S2 departures produce the systematic, target-asymmetric offsets of Proposition~3.
With-treatment OP retains the natural treatment mechanism; its reported difference from the reference-unit conversion and its bias, RMSE, and coverage relative to the scenario-specific structural/common-policy benchmark are comparator diagnostics, not performance measures for the OP estimand itself.
The recorded absolute-output failure rate is zero in every cell under the definition above.

\begingroup\scriptsize
\setlength{\tabcolsep}{2.5pt}
\begin{longtable}{llccccrrrrrrrr}
\caption{Mapping accuracy without treatment (observed practice), all scenarios, grid-averaged.
Unsuffixed performance summaries are averaged only over regular interior-supported inversion points.
Supp is the number of regular interior-supported points over the requested grid size.
Bias, RMSE, empirical SE (empSE), mean reported SE (meanSE) and band width are on the mapping scale; Cov$_{\rm pt}$/Cov$_{\rm sim}$ are interior-supported pointwise/simultaneous coverage (\%).
When a target score lies below the reference range, the operational output is clamped at $X=0$; T-clamp is the fraction of grid truths requiring this fallback.
An attainable boundary root can define the mapping but is excluded from regular performance summaries.
The spline learner SL is evaluated only at $n\in\{1000,2000\}$.}
\label{tab:sim-comp-obs}\\
\hline
Scn & Learner & Target & $m_{\rmref}$ & $n$ & Supp & Bias & RMSE & empSE & meanSE & Cov$_{\rm pt}$ & Cov$_{\rm sim}$ & Width & T-clamp \\
\hline\endfirsthead
\multicolumn{14}{l}{\small\itshape continued}\\ \hline
Scn & Learner & Target & $m_{\rmref}$ & $n$ & Supp & Bias & RMSE & empSE & meanSE & Cov$_{\rm pt}$ & Cov$_{\rm sim}$ & Width & T-clamp \\
\hline\endhead
S1a & LL & abs & 0 & 500 & 5/5 & -0.007 & 0.249 & 0.248 & 0.237 & 96 & 96 & 0.93 & 0 \\
 & LL & abs & 0 & 1000 & 5/5 & -0.001 & 0.164 & 0.164 & 0.160 & 95 & 95 & 0.63 & 0 \\
 & LL & abs & 0 & 2000 & 5/5 & 0.003 & 0.110 & 0.110 & 0.110 & 96 & 95 & 0.43 & 0 \\
 & LL & origin & 0 & 500 & 5/5 & 0.058 & 0.411 & 0.407 & 0.375 & 95 & 95 & 1.47 & 0 \\
 & LL & origin & 0 & 1000 & 5/5 & 0.028 & 0.256 & 0.255 & 0.250 & 96 & 96 & 0.98 & 0 \\
 & LL & origin & 0 & 2000 & 5/5 & 0.002 & 0.169 & 0.169 & 0.170 & 96 & 96 & 0.67 & 0 \\
S1b & LL & abs & 0 & 500 & 5/5 & 0.064 & 0.310 & 0.293 & 0.287 & 94 & 92 & 1.13 & 0 \\
 & LL & abs & 0 & 1000 & 5/5 & 0.056 & 0.219 & 0.197 & 0.194 & 93 & 88 & 0.76 & 0 \\
 & LL & abs & 0 & 2000 & 5/5 & 0.049 & 0.165 & 0.137 & 0.135 & 90 & 82 & 0.53 & 0 \\
 & LL & origin & 0 & 500 & 5/5 & -0.185 & 0.534 & 0.475 & 0.480 & 85 & 73 & 1.88 & 0 \\
 & LL & origin & 0 & 1000 & 5/5 & -0.219 & 0.425 & 0.335 & 0.321 & 74 & 50 & 1.26 & 0 \\
 & LL & origin & 0 & 2000 & 5/5 & -0.230 & 0.346 & 0.227 & 0.223 & 62 & 28 & 0.88 & 0 \\
 & SL & abs & 0 & 1000 & 5/5 & -0.082 & 0.394 & 0.375 & 0.376 & 96 & 94 & 1.38 & 0 \\
 & SL & abs & 0 & 2000 & 5/5 & -0.085 & 0.307 & 0.291 & 0.292 & 94 & 92 & 1.06 & 0 \\
 & SL & origin & 0 & 1000 & 5/5 & 0.255 & 0.924 & 0.886 & 0.835 & 97 & 96 & 2.88 & 0 \\
 & SL & origin & 0 & 2000 & 5/5 & 0.106 & 0.699 & 0.689 & 0.681 & 97 & 97 & 2.48 & 0 \\
S2a & LL & abs & 0 & 500 & 3/5 & 0.006 & 0.212 & 0.212 & 0.231 & 96 & 96 & 0.80 & 40 \\
 & LL & abs & 0 & 1000 & 3/5 & 0.002 & 0.156 & 0.156 & 0.156 & 95 & 96 & 0.57 & 40 \\
 & LL & abs & 0 & 2000 & 3/5 & 0.002 & 0.112 & 0.112 & 0.109 & 94 & 96 & 0.41 & 40 \\
 & LL & origin & 0 & 500 & 5/5 & 0.065 & 0.405 & 0.400 & 0.384 & 95 & 95 & 1.49 & 0 \\
 & LL & origin & 0 & 1000 & 5/5 & 0.021 & 0.252 & 0.251 & 0.254 & 96 & 97 & 0.99 & 0 \\
 & LL & origin & 0 & 2000 & 5/5 & 0.012 & 0.181 & 0.181 & 0.176 & 95 & 96 & 0.69 & 0 \\
S2b & LL & abs & 0 & 500 & 5/5 & 0.008 & 0.222 & 0.222 & 0.243 & 95 & 95 & 0.81 & 0 \\
 & LL & abs & 0 & 1000 & 5/5 & -0.004 & 0.153 & 0.152 & 0.163 & 96 & 95 & 0.58 & 0 \\
 & LL & abs & 0 & 2000 & 5/5 & -0.000 & 0.109 & 0.109 & 0.113 & 95 & 95 & 0.42 & 0 \\
 & LL & origin & 0 & 500 & 5/5 & 0.042 & 0.363 & 0.360 & 0.349 & 96 & 96 & 1.29 & 0 \\
 & LL & origin & 0 & 1000 & 5/5 & 0.021 & 0.241 & 0.240 & 0.232 & 96 & 96 & 0.90 & 0 \\
 & LL & origin & 0 & 2000 & 5/5 & 0.006 & 0.162 & 0.162 & 0.160 & 95 & 95 & 0.63 & 0 \\
S2c & LL & abs & 0 & 500 & 5/5 & 0.020 & 0.216 & 0.213 & 0.248 & 96 & 95 & 0.79 & 0 \\
 & LL & abs & 0 & 1000 & 5/5 & 0.012 & 0.153 & 0.152 & 0.170 & 96 & 96 & 0.57 & 0 \\
 & LL & abs & 0 & 2000 & 5/5 & 0.002 & 0.109 & 0.108 & 0.118 & 95 & 96 & 0.41 & 0 \\
 & LL & origin & 0 & 500 & 5/5 & 0.044 & 0.384 & 0.381 & 0.368 & 94 & 95 & 1.43 & 0 \\
 & LL & origin & 0 & 1000 & 5/5 & 0.030 & 0.256 & 0.255 & 0.249 & 96 & 96 & 0.98 & 0 \\
 & LL & origin & 0 & 2000 & 5/5 & 0.009 & 0.174 & 0.174 & 0.172 & 96 & 96 & 0.67 & 0 \\
S3a & LL & abs & 0 & 500 & 15/15 & 0.003 & 0.300 & 0.296 & 0.280 & 95 & 92 & 1.10 & 0 \\
 & LL & abs & 0 & 1000 & 15/15 & 0.002 & 0.164 & 0.163 & 0.162 & 96 & 94 & 0.63 & 0 \\
 & LL & abs & 0 & 2000 & 15/15 & 0.000 & 0.106 & 0.105 & 0.107 & 96 & 96 & 0.42 & 0 \\
 & LL & abs & 0.5 & 500 & 15/15 & 0.000 & 0.103 & 0.102 & 0.104 & 98 & 97 & 0.41 & 0 \\
 & LL & abs & 0.5 & 1000 & 15/15 & 0.000 & 0.069 & 0.069 & 0.070 & 97 & 96 & 0.27 & 0 \\
 & LL & abs & 0.5 & 2000 & 15/15 & 0.000 & 0.047 & 0.047 & 0.048 & 97 & 96 & 0.19 & 0 \\
 & LL & abs & 1 & 500 & 15/15 & 0.401 & 3.617 & 3.586 & 8.397 & 94 & 92 & 32.92 & 0 \\
 & LL & abs & 1 & 1000 & 15/15 & 0.172 & 2.085 & 2.076 & 2.407 & 94 & 92 & 9.43 & 0 \\
 & LL & abs & 1 & 2000 & 15/15 & 0.039 & 0.281 & 0.277 & 0.270 & 94 & 94 & 1.06 & 0 \\
 & LL & origin & 0 & 500 & 15/15 & 0.111 & 0.560 & 0.550 & 0.512 & 93 & 93 & 2.01 & 0 \\
 & LL & origin & 0 & 1000 & 15/15 & 0.050 & 0.289 & 0.285 & 0.282 & 94 & 95 & 1.10 & 0 \\
 & LL & origin & 0 & 2000 & 15/15 & 0.028 & 0.184 & 0.182 & 0.183 & 96 & 96 & 0.72 & 0 \\
 & LL & origin & 0.5 & 500 & 15/15 & 0.000 & 0.175 & 0.175 & 0.178 & 98 & 97 & 0.70 & 0 \\
 & LL & origin & 0.5 & 1000 & 15/15 & 0.000 & 0.117 & 0.117 & 0.120 & 98 & 97 & 0.47 & 0 \\
 & LL & origin & 0.5 & 2000 & 15/15 & 0.000 & 0.081 & 0.081 & 0.083 & 97 & 96 & 0.32 & 0 \\
 & LL & origin & 1 & 500 & 15/15 & 0.911 & 6.538 & 6.478 & 15.497 & 91 & 92 & 60.75 & 0 \\
 & LL & origin & 1 & 1000 & 15/15 & 0.382 & 4.145 & 4.130 & 4.710 & 93 & 93 & 18.46 & 0 \\
 & LL & origin & 1 & 2000 & 15/15 & 0.080 & 0.475 & 0.469 & 0.453 & 92 & 93 & 1.77 & 0 \\
 & SL & abs & 0 & 1000 & 15/15 & 0.017 & 0.368 & 0.357 & 0.352 & 97 & 96 & 1.30 & 0 \\
 & SL & abs & 0 & 2000 & 15/15 & 0.009 & 0.288 & 0.280 & 0.282 & 96 & 94 & 1.02 & 0 \\
 & SL & abs & 0.5 & 1000 & 15/15 & -0.002 & 0.151 & 0.150 & 0.169 & 97 & 94 & 0.64 & 0 \\
 & SL & abs & 0.5 & 2000 & 15/15 & -0.002 & 0.100 & 0.100 & 0.122 & 97 & 91 & 0.46 & 0 \\
 & SL & abs & 1 & 1000 & 15/15 & 0.091 & 0.581 & 0.559 & 0.482 & 93 & 64 & 1.70 & 0 \\
 & SL & abs & 1 & 2000 & 15/15 & 0.133 & 0.635 & 0.611 & 0.477 & 94 & 83 & 1.62 & 0 \\
 & SL & origin & 0 & 1000 & 15/15 & 0.071 & 0.679 & 0.669 & 0.609 & 97 & 99 & 2.16 & 0 \\
 & SL & origin & 0 & 2000 & 15/15 & 0.047 & 0.571 & 0.566 & 0.522 & 97 & 98 & 1.90 & 0 \\
 & SL & origin & 0.5 & 1000 & 15/15 & 0.042 & 0.341 & 0.338 & 0.351 & 98 & 97 & 1.30 & 0 \\
 & SL & origin & 0.5 & 2000 & 15/15 & 0.020 & 0.254 & 0.254 & 0.270 & 98 & 96 & 1.02 & 0 \\
 & SL & origin & 1 & 1000 & 15/15 & -0.224 & 0.906 & 0.862 & 0.755 & 92 & 64 & 2.57 & 0 \\
 & SL & origin & 1 & 2000 & 15/15 & -0.024 & 0.817 & 0.809 & 0.724 & 95 & 86 & 2.51 & 0 \\
S3b & LL & abs & 0 & 500 & 15/15 & 0.055 & 0.433 & 0.425 & 0.350 & 95 & 91 & 1.37 & 0 \\
 & LL & abs & 0 & 1000 & 15/15 & 0.027 & 0.164 & 0.156 & 0.151 & 96 & 92 & 0.59 & 0 \\
 & LL & abs & 0 & 2000 & 15/15 & 0.029 & 0.115 & 0.104 & 0.101 & 94 & 91 & 0.40 & 0 \\
 & LL & origin & 0 & 500 & 15/15 & -0.003 & 0.831 & 0.827 & 0.639 & 82 & 73 & 2.51 & 0 \\
 & LL & origin & 0 & 1000 & 15/15 & -0.082 & 0.294 & 0.271 & 0.251 & 79 & 64 & 0.98 & 0 \\
 & LL & origin & 0 & 2000 & 15/15 & -0.098 & 0.212 & 0.173 & 0.165 & 73 & 52 & 0.65 & 0 \\
 & SL & abs & 0 & 1000 & 15/15 & 0.000 & 0.350 & 0.342 & 0.350 & 97 & 96 & 1.32 & 0 \\
 & SL & abs & 0 & 2000 & 15/15 & -0.016 & 0.234 & 0.228 & 0.266 & 96 & 95 & 0.98 & 0 \\
 & SL & origin & 0 & 1000 & 15/15 & -0.052 & 0.628 & 0.621 & 0.597 & 98 & 98 & 2.14 & 0 \\
 & SL & origin & 0 & 2000 & 15/15 & 0.005 & 0.526 & 0.523 & 0.514 & 97 & 99 & 1.88 & 0 \\
\hline
\end{longtable}
\endgroup
\begingroup\scriptsize
\setlength{\tabcolsep}{1.4pt}
\begin{longtable}{lllllcccrrrrr}
\caption{Difference from the reference-unit conversion, all scenarios $\times$ learner $\times$ axis $\times$ policy $\times$ target $\times$ $n$ ($\times\,m_{\rmref}$); grid-averaged.
The signed supported-point difference is the grid average, over regular interior-supported inversion points, of $\tfrac1R\sum_r\hat L_k^{(r)}-L_k^{\rm rec}$, where $L^{\rm rec}=g_{m_{\rmref}}\!\circ g_m^{-1}$ (Proposition~3); the corresponding absolute difference averages the absolute value of this point-specific quantity.
The signed full-grid difference instead averages $\hat L_k^{\rm op}-L_k^{\rm rec}$ over the full requested grid after the prespecified $X=0$ lower-bound fallback.
For the no-treatment OP mapping and the PS-stc mapping, G1--G4 (S1, S3) imply that the difference is approximately zero (the prognosis-equivalent mapping coincides with the reference-unit conversion); for those same targets, the targeted departures (S2) produce systematic recovery offsets that appear asymmetrically across the absolute (abs) and origin-referenced (origin) targets.
Axis is without (no-trt) or with (trt) treatment.
With-treatment OP retains the natural treatment mechanism: its reported difference here, and its bias, RMSE, and coverage in Table~\ref{tab:sim-comp-treat}, are comparator diagnostics relative to the scenario-specific structural/common-policy benchmark, not bias, RMSE, or coverage for the OP estimand itself.
PS-stc is the static policy-standardized mapping.
Supp and Cov$_{\rm pt}$ refer to the regular interior-supported range; an attainable boundary root can define the mapping but is excluded from regular performance summaries.
T-clamp is the percentage of grid truths requiring the operational fallback.
A dash in the full-grid difference denotes a legacy result object that predates lower-bound projection and did not record that diagnostic.
The spline learner SL is at $n\in\{1000,2000\}$.}
\label{tab:sim-comp-recovery}\\
\hline
Scn & Learner & Axis & Policy & Target & $m_{\rmref}$ & $n$ & Supp & \multicolumn{3}{c}{Difference from $L^{\rm rec}$} & Cov$_{\rm pt}$ & T-clamp \\
 & & & & & & & & \shortstack{Mean signed\\(supported)} & \shortstack{Mean absolute\\(supported)} & \shortstack{Mean signed\\(full grid,\\with fallback)} & & \\
\hline\endfirsthead
\multicolumn{13}{l}{\small\itshape continued}\\ \hline
Scn & Learner & Axis & Policy & Target & $m_{\rmref}$ & $n$ & Supp & \multicolumn{3}{c}{Difference from $L^{\rm rec}$} & Cov$_{\rm pt}$ & T-clamp \\
 & & & & & & & & \shortstack{Mean signed\\(supported)} & \shortstack{Mean absolute\\(supported)} & \shortstack{Mean signed\\(full grid,\\with fallback)} & & \\
\hline\endhead
S1a & LL & no-trt & OP & abs & 0 & 500 & 5/5 & -0.007 & 0.024 & --- & 96 & 0 \\
 & LL & no-trt & OP & abs & 0 & 1000 & 5/5 & -0.001 & 0.012 & --- & 95 & 0 \\
 & LL & no-trt & OP & abs & 0 & 2000 & 5/5 & 0.003 & 0.003 & --- & 96 & 0 \\
 & LL & no-trt & OP & origin & 0 & 500 & 5/5 & 0.058 & 0.058 & --- & 95 & 0 \\
 & LL & no-trt & OP & origin & 0 & 1000 & 5/5 & 0.028 & 0.028 & --- & 96 & 0 \\
 & LL & no-trt & OP & origin & 0 & 2000 & 5/5 & 0.002 & 0.002 & --- & 96 & 0 \\
 & LL & trt & OP & abs & 0 & 500 & 5/5 & -0.210 & 0.210 & --- & 90 & 0 \\
 & LL & trt & OP & abs & 0 & 1000 & 5/5 & -0.206 & 0.206 & --- & 79 & 0 \\
 & LL & trt & OP & abs & 0 & 2000 & 5/5 & -0.204 & 0.204 & --- & 60 & 0 \\
 & LL & trt & OP & origin & 0 & 500 & 5/5 & 0.038 & 0.038 & --- & 96 & 0 \\
 & LL & trt & OP & origin & 0 & 1000 & 5/5 & 0.003 & 0.003 & --- & 96 & 0 \\
 & LL & trt & OP & origin & 0 & 2000 & 5/5 & -0.005 & 0.005 & --- & 95 & 0 \\
 & LL & trt & PS-stc & abs & 0 & 500 & 5/5 & -0.012 & 0.027 & --- & 97 & 0 \\
 & LL & trt & PS-stc & abs & 0 & 1000 & 5/5 & -0.004 & 0.011 & --- & 96 & 0 \\
 & LL & trt & PS-stc & abs & 0 & 2000 & 5/5 & -0.001 & 0.007 & --- & 96 & 0 \\
 & LL & trt & PS-stc & origin & 0 & 500 & 5/5 & 0.061 & 0.061 & --- & 96 & 0 \\
 & LL & trt & PS-stc & origin & 0 & 1000 & 5/5 & 0.026 & 0.026 & --- & 97 & 0 \\
 & LL & trt & PS-stc & origin & 0 & 2000 & 5/5 & 0.017 & 0.017 & --- & 96 & 0 \\
S1b & LL & no-trt & OP & abs & 0 & 500 & 5/5 & 0.064 & 0.095 & --- & 94 & 0 \\
 & LL & no-trt & OP & abs & 0 & 1000 & 5/5 & 0.056 & 0.089 & --- & 93 & 0 \\
 & LL & no-trt & OP & abs & 0 & 2000 & 5/5 & 0.049 & 0.086 & --- & 90 & 0 \\
 & LL & no-trt & OP & origin & 0 & 500 & 5/5 & -0.185 & 0.187 & --- & 85 & 0 \\
 & LL & no-trt & OP & origin & 0 & 1000 & 5/5 & -0.219 & 0.219 & --- & 74 & 0 \\
 & LL & no-trt & OP & origin & 0 & 2000 & 5/5 & -0.230 & 0.230 & --- & 62 & 0 \\
 & LL & trt & OP & abs & 0 & 500 & 5/5 & -0.144 & 0.162 & --- & 92 & 0 \\
 & LL & trt & OP & abs & 0 & 1000 & 5/5 & -0.149 & 0.150 & --- & 83 & 0 \\
 & LL & trt & OP & abs & 0 & 2000 & 5/5 & -0.154 & 0.154 & --- & 67 & 0 \\
 & LL & trt & OP & origin & 0 & 500 & 5/5 & -0.187 & 0.188 & --- & 84 & 0 \\
 & LL & trt & OP & origin & 0 & 1000 & 5/5 & -0.241 & 0.241 & --- & 77 & 0 \\
 & LL & trt & OP & origin & 0 & 2000 & 5/5 & -0.288 & 0.288 & --- & 60 & 0 \\
 & LL & trt & PS-stc & abs & 0 & 500 & 5/5 & 0.025 & 0.094 & --- & 95 & 0 \\
 & LL & trt & PS-stc & abs & 0 & 1000 & 5/5 & 0.048 & 0.090 & --- & 95 & 0 \\
 & LL & trt & PS-stc & abs & 0 & 2000 & 5/5 & 0.041 & 0.082 & --- & 91 & 0 \\
 & LL & trt & PS-stc & origin & 0 & 500 & 5/5 & -0.136 & 0.170 & --- & 87 & 0 \\
 & LL & trt & PS-stc & origin & 0 & 1000 & 5/5 & -0.200 & 0.200 & --- & 79 & 0 \\
 & LL & trt & PS-stc & origin & 0 & 2000 & 5/5 & -0.250 & 0.250 & --- & 62 & 0 \\
 & SL & no-trt & OP & abs & 0 & 1000 & 5/5 & -0.082 & 0.114 & --- & 96 & 0 \\
 & SL & no-trt & OP & abs & 0 & 2000 & 5/5 & -0.085 & 0.085 & --- & 94 & 0 \\
 & SL & no-trt & OP & origin & 0 & 1000 & 5/5 & 0.255 & 0.255 & --- & 97 & 0 \\
 & SL & no-trt & OP & origin & 0 & 2000 & 5/5 & 0.106 & 0.106 & --- & 97 & 0 \\
 & SL & trt & OP & abs & 0 & 1000 & 5/5 & -0.306 & 0.306 & --- & 84 & 0 \\
 & SL & trt & OP & abs & 0 & 2000 & 5/5 & -0.296 & 0.296 & --- & 75 & 0 \\
 & SL & trt & OP & origin & 0 & 1000 & 5/5 & 0.309 & 0.309 & --- & 97 & 0 \\
 & SL & trt & OP & origin & 0 & 2000 & 5/5 & 0.255 & 0.255 & --- & 96 & 0 \\
 & SL & trt & PS-stc & abs & 0 & 1000 & 5/5 & -0.113 & 0.152 & --- & 97 & 0 \\
 & SL & trt & PS-stc & abs & 0 & 2000 & 5/5 & -0.101 & 0.118 & --- & 94 & 0 \\
 & SL & trt & PS-stc & origin & 0 & 1000 & 5/5 & 0.290 & 0.290 & --- & 99 & 0 \\
 & SL & trt & PS-stc & origin & 0 & 2000 & 5/5 & 0.245 & 0.245 & --- & 98 & 0 \\
S2a & LL & no-trt & OP & abs & 0 & 500 & 3/5 & -0.709 & 0.709 & -0.604 & 96 & 40 \\
 & LL & no-trt & OP & abs & 0 & 1000 & 3/5 & -0.712 & 0.712 & -0.615 & 95 & 40 \\
 & LL & no-trt & OP & abs & 0 & 2000 & 3/5 & -0.712 & 0.712 & -0.619 & 94 & 40 \\
 & LL & no-trt & OP & origin & 0 & 500 & 5/5 & 0.065 & 0.065 & 0.065 & 95 & 0 \\
 & LL & no-trt & OP & origin & 0 & 1000 & 5/5 & 0.021 & 0.021 & 0.021 & 96 & 0 \\
 & LL & no-trt & OP & origin & 0 & 2000 & 5/5 & 0.012 & 0.012 & 0.012 & 95 & 0 \\
 & LL & trt & OP & abs & 0 & 500 & 3/5 & -0.935 & 0.935 & -0.754 & 95 & 40 \\
 & LL & trt & OP & abs & 0 & 1000 & 3/5 & -0.947 & 0.947 & -0.766 & 81 & 40 \\
 & LL & trt & OP & abs & 0 & 2000 & 3/5 & -0.950 & 0.950 & -0.770 & 56 & 40 \\
 & LL & trt & OP & origin & 0 & 500 & 5/5 & 0.041 & 0.041 & 0.041 & 95 & 0 \\
 & LL & trt & OP & origin & 0 & 1000 & 5/5 & 0.003 & 0.003 & 0.003 & 95 & 0 \\
 & LL & trt & OP & origin & 0 & 2000 & 5/5 & -0.004 & 0.004 & -0.004 & 94 & 0 \\
 & LL & trt & PS-stc & abs & 0 & 500 & 3/5 & -0.718 & 0.718 & -0.602 & 94 & 40 \\
 & LL & trt & PS-stc & abs & 0 & 1000 & 3/5 & -0.720 & 0.720 & -0.615 & 95 & 40 \\
 & LL & trt & PS-stc & abs & 0 & 2000 & 3/5 & -0.717 & 0.717 & -0.620 & 95 & 40 \\
 & LL & trt & PS-stc & origin & 0 & 500 & 5/5 & 0.078 & 0.078 & 0.078 & 96 & 0 \\
 & LL & trt & PS-stc & origin & 0 & 1000 & 5/5 & 0.033 & 0.033 & 0.033 & 96 & 0 \\
 & LL & trt & PS-stc & origin & 0 & 2000 & 5/5 & 0.019 & 0.019 & 0.019 & 96 & 0 \\
S2b & LL & no-trt & OP & abs & 0 & 500 & 5/5 & -0.277 & 0.277 & -0.277 & 95 & 0 \\
 & LL & no-trt & OP & abs & 0 & 1000 & 5/5 & -0.290 & 0.290 & -0.290 & 96 & 0 \\
 & LL & no-trt & OP & abs & 0 & 2000 & 5/5 & -0.286 & 0.286 & -0.286 & 95 & 0 \\
 & LL & no-trt & OP & origin & 0 & 500 & 5/5 & -0.244 & 0.244 & -0.244 & 96 & 0 \\
 & LL & no-trt & OP & origin & 0 & 1000 & 5/5 & -0.264 & 0.264 & -0.264 & 96 & 0 \\
 & LL & no-trt & OP & origin & 0 & 2000 & 5/5 & -0.280 & 0.280 & -0.280 & 95 & 0 \\
 & LL & trt & OP & abs & 0 & 500 & 5/5 & -0.535 & 0.535 & --- & 94 & 0 \\
 & LL & trt & OP & abs & 0 & 1000 & 5/5 & -0.518 & 0.518 & --- & 80 & 0 \\
 & LL & trt & OP & abs & 0 & 2000 & 5/5 & -0.506 & 0.506 & --- & 61 & 0 \\
 & LL & trt & OP & origin & 0 & 500 & 5/5 & -0.275 & 0.275 & --- & 96 & 0 \\
 & LL & trt & OP & origin & 0 & 1000 & 5/5 & -0.311 & 0.311 & --- & 95 & 0 \\
 & LL & trt & OP & origin & 0 & 2000 & 5/5 & -0.312 & 0.312 & --- & 95 & 0 \\
 & LL & trt & PS-stc & abs & 0 & 500 & 5/5 & -0.318 & 0.318 & --- & 97 & 0 \\
 & LL & trt & PS-stc & abs & 0 & 1000 & 5/5 & -0.301 & 0.301 & --- & 96 & 0 \\
 & LL & trt & PS-stc & abs & 0 & 2000 & 5/5 & -0.291 & 0.291 & --- & 96 & 0 \\
 & LL & trt & PS-stc & origin & 0 & 500 & 5/5 & -0.254 & 0.254 & --- & 97 & 0 \\
 & LL & trt & PS-stc & origin & 0 & 1000 & 5/5 & -0.273 & 0.273 & --- & 96 & 0 \\
 & LL & trt & PS-stc & origin & 0 & 2000 & 5/5 & -0.274 & 0.274 & --- & 96 & 0 \\
S2c & LL & no-trt & OP & abs & 0 & 500 & 5/5 & 0.020 & 0.025 & 0.020 & 96 & 0 \\
 & LL & no-trt & OP & abs & 0 & 1000 & 5/5 & 0.012 & 0.016 & 0.012 & 96 & 0 \\
 & LL & no-trt & OP & abs & 0 & 2000 & 5/5 & 0.002 & 0.006 & 0.002 & 95 & 0 \\
 & LL & no-trt & OP & origin & 0 & 500 & 5/5 & 0.244 & 0.244 & 0.244 & 94 & 0 \\
 & LL & no-trt & OP & origin & 0 & 1000 & 5/5 & 0.230 & 0.230 & 0.230 & 96 & 0 \\
 & LL & no-trt & OP & origin & 0 & 2000 & 5/5 & 0.209 & 0.209 & 0.209 & 96 & 0 \\
 & LL & trt & OP & abs & 0 & 500 & 5/5 & -0.172 & 0.172 & -0.172 & 90 & 0 \\
 & LL & trt & OP & abs & 0 & 1000 & 5/5 & -0.187 & 0.187 & -0.187 & 77 & 0 \\
 & LL & trt & OP & abs & 0 & 2000 & 5/5 & -0.189 & 0.189 & -0.189 & 58 & 0 \\
 & LL & trt & OP & origin & 0 & 500 & 5/5 & 0.242 & 0.242 & 0.242 & 94 & 0 \\
 & LL & trt & OP & origin & 0 & 1000 & 5/5 & 0.216 & 0.216 & 0.216 & 95 & 0 \\
 & LL & trt & OP & origin & 0 & 2000 & 5/5 & 0.187 & 0.187 & 0.187 & 95 & 0 \\
 & LL & trt & PS-stc & abs & 0 & 500 & 5/5 & 0.013 & 0.031 & 0.013 & 96 & 0 \\
 & LL & trt & PS-stc & abs & 0 & 1000 & 5/5 & 0.003 & 0.019 & 0.003 & 96 & 0 \\
 & LL & trt & PS-stc & abs & 0 & 2000 & 5/5 & 0.008 & 0.009 & 0.008 & 95 & 0 \\
 & LL & trt & PS-stc & origin & 0 & 500 & 5/5 & 0.264 & 0.264 & 0.264 & 96 & 0 \\
 & LL & trt & PS-stc & origin & 0 & 1000 & 5/5 & 0.236 & 0.236 & 0.236 & 96 & 0 \\
 & LL & trt & PS-stc & origin & 0 & 2000 & 5/5 & 0.207 & 0.207 & 0.207 & 96 & 0 \\
S3a & LL & no-trt & OP & abs & 0 & 500 & 15/15 & 0.003 & 0.045 & --- & 95 & 0 \\
 & LL & no-trt & OP & abs & 0 & 1000 & 15/15 & 0.002 & 0.020 & --- & 96 & 0 \\
 & LL & no-trt & OP & abs & 0 & 2000 & 15/15 & 0.000 & 0.011 & --- & 96 & 0 \\
 & LL & no-trt & OP & abs & 0.5 & 500 & 15/15 & 0.000 & 0.006 & --- & 98 & 0 \\
 & LL & no-trt & OP & abs & 0.5 & 1000 & 15/15 & 0.000 & 0.002 & --- & 97 & 0 \\
 & LL & no-trt & OP & abs & 0.5 & 2000 & 15/15 & 0.000 & 0.001 & --- & 97 & 0 \\
 & LL & no-trt & OP & abs & 1 & 500 & 15/15 & 0.401 & 0.484 & --- & 94 & 0 \\
 & LL & no-trt & OP & abs & 1 & 1000 & 15/15 & 0.172 & 0.205 & --- & 94 & 0 \\
 & LL & no-trt & OP & abs & 1 & 2000 & 15/15 & 0.039 & 0.045 & --- & 94 & 0 \\
 & LL & no-trt & OP & origin & 0 & 500 & 15/15 & 0.111 & 0.111 & --- & 93 & 0 \\
 & LL & no-trt & OP & origin & 0 & 1000 & 15/15 & 0.050 & 0.050 & --- & 94 & 0 \\
 & LL & no-trt & OP & origin & 0 & 2000 & 15/15 & 0.028 & 0.028 & --- & 96 & 0 \\
 & LL & no-trt & OP & origin & 0.5 & 500 & 15/15 & 0.000 & 0.010 & --- & 98 & 0 \\
 & LL & no-trt & OP & origin & 0.5 & 1000 & 15/15 & 0.000 & 0.001 & --- & 98 & 0 \\
 & LL & no-trt & OP & origin & 0.5 & 2000 & 15/15 & 0.000 & 0.002 & --- & 97 & 0 \\
 & LL & no-trt & OP & origin & 1 & 500 & 15/15 & 0.911 & 0.911 & --- & 91 & 0 \\
 & LL & no-trt & OP & origin & 1 & 1000 & 15/15 & 0.382 & 0.382 & --- & 93 & 0 \\
 & LL & no-trt & OP & origin & 1 & 2000 & 15/15 & 0.080 & 0.080 & --- & 92 & 0 \\
 & LL & trt & OP & abs & 0 & 500 & 15/15 & -0.095 & 0.095 & --- & 94 & 0 \\
 & LL & trt & OP & abs & 0 & 1000 & 15/15 & -0.104 & 0.104 & --- & 92 & 0 \\
 & LL & trt & OP & abs & 0 & 2000 & 15/15 & -0.107 & 0.107 & --- & 84 & 0 \\
 & LL & trt & OP & abs & 0.5 & 500 & 15/15 & 0.000 & 0.046 & --- & 97 & 0 \\
 & LL & trt & OP & abs & 0.5 & 1000 & 15/15 & 0.000 & 0.043 & --- & 95 & 0 \\
 & LL & trt & OP & abs & 0.5 & 2000 & 15/15 & 0.000 & 0.046 & --- & 90 & 0 \\
 & LL & trt & OP & abs & 1 & 500 & 15/15 & 0.504 & 0.504 & --- & 93 & 0 \\
 & LL & trt & OP & abs & 1 & 1000 & 15/15 & 0.329 & 0.329 & --- & 93 & 0 \\
 & LL & trt & OP & abs & 1 & 2000 & 15/15 & 0.292 & 0.292 & --- & 93 & 0 \\
 & LL & trt & OP & origin & 0 & 500 & 15/15 & 0.032 & 0.032 & --- & 92 & 0 \\
 & LL & trt & OP & origin & 0 & 1000 & 15/15 & 0.081 & 0.081 & --- & 93 & 0 \\
 & LL & trt & OP & origin & 0 & 2000 & 15/15 & 0.008 & 0.008 & --- & 92 & 0 \\
 & LL & trt & OP & origin & 0.5 & 500 & 15/15 & 0.000 & 0.005 & --- & 97 & 0 \\
 & LL & trt & OP & origin & 0.5 & 1000 & 15/15 & 0.000 & 0.001 & --- & 97 & 0 \\
 & LL & trt & OP & origin & 0.5 & 2000 & 15/15 & 0.000 & 0.012 & --- & 97 & 0 \\
 & LL & trt & OP & origin & 1 & 500 & 15/15 & 0.753 & 0.753 & --- & 88 & 0 \\
 & LL & trt & OP & origin & 1 & 1000 & 15/15 & 0.318 & 0.318 & --- & 90 & 0 \\
 & LL & trt & OP & origin & 1 & 2000 & 15/15 & 0.270 & 0.270 & --- & 95 & 0 \\
 & LL & trt & PS-stc & abs & 0 & 500 & 15/15 & 0.124 & 0.286 & --- & 96 & 0 \\
 & LL & trt & PS-stc & abs & 0 & 1000 & 15/15 & 0.004 & 0.066 & --- & 97 & 0 \\
 & LL & trt & PS-stc & abs & 0 & 2000 & 15/15 & -0.004 & 0.009 & --- & 95 & 0 \\
 & LL & trt & PS-stc & abs & 0.5 & 500 & 15/15 & 0.000 & 0.003 & --- & 98 & 0 \\
 & LL & trt & PS-stc & abs & 0.5 & 1000 & 15/15 & 0.000 & 0.004 & --- & 98 & 0 \\
 & LL & trt & PS-stc & abs & 0.5 & 2000 & 15/15 & 0.000 & 0.003 & --- & 97 & 0 \\
 & LL & trt & PS-stc & abs & 1 & 500 & 15/15 & 0.218 & 0.218 & --- & 93 & 0 \\
 & LL & trt & PS-stc & abs & 1 & 1000 & 15/15 & 0.039 & 0.049 & --- & 93 & 0 \\
 & LL & trt & PS-stc & abs & 1 & 2000 & 15/15 & 0.058 & 0.064 & --- & 95 & 0 \\
 & LL & trt & PS-stc & origin & 0 & 500 & 15/15 & 0.653 & 0.653 & --- & 94 & 0 \\
 & LL & trt & PS-stc & origin & 0 & 1000 & 15/15 & 0.162 & 0.162 & --- & 95 & 0 \\
 & LL & trt & PS-stc & origin & 0 & 2000 & 15/15 & 0.020 & 0.020 & --- & 93 & 0 \\
 & LL & trt & PS-stc & origin & 0.5 & 500 & 15/15 & 0.000 & 0.006 & --- & 98 & 0 \\
 & LL & trt & PS-stc & origin & 0.5 & 1000 & 15/15 & 0.000 & 0.010 & --- & 98 & 0 \\
 & LL & trt & PS-stc & origin & 0.5 & 2000 & 15/15 & 0.000 & 0.005 & --- & 97 & 0 \\
 & LL & trt & PS-stc & origin & 1 & 500 & 15/15 & 0.036 & 0.036 & --- & 90 & 0 \\
 & LL & trt & PS-stc & origin & 1 & 1000 & 15/15 & 0.095 & 0.095 & --- & 89 & 0 \\
 & LL & trt & PS-stc & origin & 1 & 2000 & 15/15 & 0.109 & 0.109 & --- & 94 & 0 \\
 & SL & no-trt & OP & abs & 0 & 1000 & 15/15 & 0.017 & 0.086 & --- & 97 & 0 \\
 & SL & no-trt & OP & abs & 0 & 2000 & 15/15 & 0.009 & 0.065 & --- & 96 & 0 \\
 & SL & no-trt & OP & abs & 0.5 & 1000 & 15/15 & -0.002 & 0.012 & --- & 97 & 0 \\
 & SL & no-trt & OP & abs & 0.5 & 2000 & 15/15 & -0.002 & 0.005 & --- & 97 & 0 \\
 & SL & no-trt & OP & abs & 1 & 1000 & 15/15 & 0.091 & 0.143 & --- & 93 & 0 \\
 & SL & no-trt & OP & abs & 1 & 2000 & 15/15 & 0.133 & 0.162 & --- & 94 & 0 \\
 & SL & no-trt & OP & origin & 0 & 1000 & 15/15 & 0.071 & 0.096 & --- & 97 & 0 \\
 & SL & no-trt & OP & origin & 0 & 2000 & 15/15 & 0.047 & 0.062 & --- & 97 & 0 \\
 & SL & no-trt & OP & origin & 0.5 & 1000 & 15/15 & 0.042 & 0.044 & --- & 98 & 0 \\
 & SL & no-trt & OP & origin & 0.5 & 2000 & 15/15 & 0.020 & 0.020 & --- & 98 & 0 \\
 & SL & no-trt & OP & origin & 1 & 1000 & 15/15 & -0.224 & 0.241 & --- & 92 & 0 \\
 & SL & no-trt & OP & origin & 1 & 2000 & 15/15 & -0.024 & 0.094 & --- & 95 & 0 \\
 & SL & trt & OP & abs & 0 & 1000 & 15/15 & -0.102 & 0.118 & --- & 96 & 0 \\
 & SL & trt & OP & abs & 0 & 2000 & 15/15 & -0.109 & 0.116 & --- & 92 & 0 \\
 & SL & trt & OP & abs & 0.5 & 1000 & 15/15 & 0.001 & 0.065 & --- & 97 & 0 \\
 & SL & trt & OP & abs & 0.5 & 2000 & 15/15 & -0.002 & 0.053 & --- & 95 & 0 \\
 & SL & trt & OP & abs & 1 & 1000 & 15/15 & 0.258 & 0.285 & --- & 92 & 0 \\
 & SL & trt & OP & abs & 1 & 2000 & 15/15 & 0.257 & 0.263 & --- & 91 & 0 \\
 & SL & trt & OP & origin & 0 & 1000 & 15/15 & 0.007 & 0.075 & --- & 98 & 0 \\
 & SL & trt & OP & origin & 0 & 2000 & 15/15 & 0.108 & 0.117 & --- & 97 & 0 \\
 & SL & trt & OP & origin & 0.5 & 1000 & 15/15 & 0.057 & 0.057 & --- & 99 & 0 \\
 & SL & trt & OP & origin & 0.5 & 2000 & 15/15 & 0.033 & 0.037 & --- & 98 & 0 \\
 & SL & trt & OP & origin & 1 & 1000 & 15/15 & -0.168 & 0.209 & --- & 94 & 0 \\
 & SL & trt & OP & origin & 1 & 2000 & 15/15 & -0.176 & 0.184 & --- & 95 & 0 \\
 & SL & trt & PS-stc & abs & 0 & 1000 & 15/15 & 0.024 & 0.068 & --- & 97 & 0 \\
 & SL & trt & PS-stc & abs & 0 & 2000 & 15/15 & 0.008 & 0.075 & --- & 96 & 0 \\
 & SL & trt & PS-stc & abs & 0.5 & 1000 & 15/15 & 0.000 & 0.020 & --- & 98 & 0 \\
 & SL & trt & PS-stc & abs & 0.5 & 2000 & 15/15 & -0.001 & 0.012 & --- & 97 & 0 \\
 & SL & trt & PS-stc & abs & 1 & 1000 & 15/15 & 0.112 & 0.162 & --- & 94 & 0 \\
 & SL & trt & PS-stc & abs & 1 & 2000 & 15/15 & 0.103 & 0.122 & --- & 95 & 0 \\
 & SL & trt & PS-stc & origin & 0 & 1000 & 15/15 & 0.001 & 0.068 & --- & 98 & 0 \\
 & SL & trt & PS-stc & origin & 0 & 2000 & 15/15 & 0.095 & 0.103 & --- & 98 & 0 \\
 & SL & trt & PS-stc & origin & 0.5 & 1000 & 15/15 & 0.049 & 0.049 & --- & 99 & 0 \\
 & SL & trt & PS-stc & origin & 0.5 & 2000 & 15/15 & 0.033 & 0.033 & --- & 99 & 0 \\
 & SL & trt & PS-stc & origin & 1 & 1000 & 15/15 & -0.123 & 0.180 & --- & 96 & 0 \\
 & SL & trt & PS-stc & origin & 1 & 2000 & 15/15 & -0.146 & 0.154 & --- & 96 & 0 \\
S3b & LL & no-trt & OP & abs & 0 & 500 & 15/15 & 0.055 & 0.065 & --- & 95 & 0 \\
 & LL & no-trt & OP & abs & 0 & 1000 & 15/15 & 0.027 & 0.048 & --- & 96 & 0 \\
 & LL & no-trt & OP & abs & 0 & 2000 & 15/15 & 0.029 & 0.046 & --- & 94 & 0 \\
 & LL & no-trt & OP & origin & 0 & 500 & 15/15 & -0.003 & 0.076 & --- & 82 & 0 \\
 & LL & no-trt & OP & origin & 0 & 1000 & 15/15 & -0.082 & 0.090 & --- & 79 & 0 \\
 & LL & no-trt & OP & origin & 0 & 2000 & 15/15 & -0.098 & 0.098 & --- & 73 & 0 \\
 & LL & trt & OP & abs & 0 & 500 & 15/15 & -0.079 & 0.113 & --- & 93 & 0 \\
 & LL & trt & OP & abs & 0 & 1000 & 15/15 & -0.081 & 0.088 & --- & 90 & 0 \\
 & LL & trt & OP & abs & 0 & 2000 & 15/15 & -0.074 & 0.076 & --- & 83 & 0 \\
 & LL & trt & OP & origin & 0 & 500 & 15/15 & 0.011 & 0.078 & --- & 81 & 0 \\
 & LL & trt & OP & origin & 0 & 1000 & 15/15 & -0.087 & 0.092 & --- & 76 & 0 \\
 & LL & trt & OP & origin & 0 & 2000 & 15/15 & -0.117 & 0.117 & --- & 70 & 0 \\
 & LL & trt & PS-stc & abs & 0 & 500 & 15/15 & 0.026 & 0.121 & --- & 97 & 0 \\
 & LL & trt & PS-stc & abs & 0 & 1000 & 15/15 & 0.014 & 0.047 & --- & 96 & 0 \\
 & LL & trt & PS-stc & abs & 0 & 2000 & 15/15 & 0.024 & 0.045 & --- & 95 & 0 \\
 & LL & trt & PS-stc & origin & 0 & 500 & 15/15 & 0.125 & 0.131 & --- & 84 & 0 \\
 & LL & trt & PS-stc & origin & 0 & 1000 & 15/15 & -0.069 & 0.086 & --- & 78 & 0 \\
 & LL & trt & PS-stc & origin & 0 & 2000 & 15/15 & -0.098 & 0.098 & --- & 72 & 0 \\
 & SL & no-trt & OP & abs & 0 & 1000 & 15/15 & 0.000 & 0.070 & --- & 97 & 0 \\
 & SL & no-trt & OP & abs & 0 & 2000 & 15/15 & -0.016 & 0.047 & --- & 96 & 0 \\
 & SL & no-trt & OP & origin & 0 & 1000 & 15/15 & -0.052 & 0.078 & --- & 98 & 0 \\
 & SL & no-trt & OP & origin & 0 & 2000 & 15/15 & 0.005 & 0.040 & --- & 97 & 0 \\
 & SL & trt & OP & abs & 0 & 1000 & 15/15 & -0.106 & 0.113 & --- & 95 & 0 \\
 & SL & trt & OP & abs & 0 & 2000 & 15/15 & -0.128 & 0.128 & --- & 93 & 0 \\
 & SL & trt & OP & origin & 0 & 1000 & 15/15 & 0.001 & 0.088 & --- & 98 & 0 \\
 & SL & trt & OP & origin & 0 & 2000 & 15/15 & -0.008 & 0.054 & --- & 97 & 0 \\
 & SL & trt & PS-stc & abs & 0 & 1000 & 15/15 & 0.017 & 0.059 & --- & 98 & 0 \\
 & SL & trt & PS-stc & abs & 0 & 2000 & 15/15 & -0.023 & 0.051 & --- & 97 & 0 \\
 & SL & trt & PS-stc & origin & 0 & 1000 & 15/15 & -0.028 & 0.081 & --- & 99 & 0 \\
 & SL & trt & PS-stc & origin & 0 & 2000 & 15/15 & -0.024 & 0.054 & --- & 98 & 0 \\
\hline
\end{longtable}
\endgroup
\begingroup\scriptsize
\setlength{\tabcolsep}{2.5pt}
\begin{longtable}{llllcccrrrrrr}
\caption{Mapping accuracy with treatment, all scenarios, grid-averaged.
OP retains the natural treatment mechanism; PS-stc is the policy-standardized mapping under the static regime assigning $A=1$ to everyone.
Target is the absolute (abs) or origin-referenced (origin) mapping.
For OP, the displayed bias, RMSE, and coverage compare with the scenario-specific structural benchmark (equivalently here, the static common-policy benchmark) and are comparator diagnostics rather than bias, RMSE, or coverage for the OP estimand itself.
Unsuffixed performance columns use only regular interior-supported inversion points.
Supp and T-clamp report the regular interior-supported fraction and truth-level operational fallback, respectively, as defined in Table~\ref{tab:sim-comp-obs}.}
\label{tab:sim-comp-treat}\\
\hline
Scn & Learner & Policy & Target & $m_{\rmref}$ & $n$ & Supp & Bias & RMSE & Cov$_{\rm pt}$ & Cov$_{\rm sim}$ & Width & T-clamp \\
\hline\endfirsthead
\multicolumn{13}{l}{\small\itshape continued}\\ \hline
Scn & Learner & Policy & Target & $m_{\rmref}$ & $n$ & Supp & Bias & RMSE & Cov$_{\rm pt}$ & Cov$_{\rm sim}$ & Width & T-clamp \\
\hline\endhead
S1a & LL & OP & abs & 0 & 500 & 5/5 & -0.210 & 0.340 & 90 & 88 & 1.06 & 0 \\
 & LL & OP & origin & 0 & 500 & 5/5 & 0.038 & 0.408 & 96 & 96 & 1.63 & 0 \\
 & LL & OP & abs & 0 & 1000 & 5/5 & -0.206 & 0.277 & 79 & 72 & 0.71 & 0 \\
 & LL & OP & origin & 0 & 1000 & 5/5 & 0.003 & 0.281 & 96 & 96 & 1.08 & 0 \\
 & LL & OP & abs & 0 & 2000 & 5/5 & -0.204 & 0.240 & 60 & 46 & 0.49 & 0 \\
 & LL & OP & origin & 0 & 2000 & 5/5 & -0.005 & 0.190 & 95 & 95 & 0.74 & 0 \\
 & LL & PS-stc & abs & 0 & 500 & 5/5 & -0.012 & 0.328 & 97 & 97 & 1.27 & 0 \\
 & LL & PS-stc & origin & 0 & 500 & 5/5 & 0.061 & 0.410 & 96 & 96 & 1.64 & 0 \\
 & LL & PS-stc & abs & 0 & 1000 & 5/5 & -0.004 & 0.209 & 96 & 97 & 0.82 & 0 \\
 & LL & PS-stc & origin & 0 & 1000 & 5/5 & 0.026 & 0.272 & 97 & 97 & 1.04 & 0 \\
 & LL & PS-stc & abs & 0 & 2000 & 5/5 & -0.001 & 0.140 & 96 & 96 & 0.56 & 0 \\
 & LL & PS-stc & origin & 0 & 2000 & 5/5 & 0.017 & 0.178 & 96 & 96 & 0.70 & 0 \\
S1b & LL & OP & abs & 0 & 500 & 5/5 & -0.144 & 0.408 & 92 & 92 & 1.33 & 0 \\
 & LL & OP & origin & 0 & 500 & 5/5 & -0.187 & 0.652 & 84 & 73 & 2.20 & 0 \\
 & LL & OP & abs & 0 & 1000 & 5/5 & -0.149 & 0.298 & 83 & 75 & 0.88 & 0 \\
 & LL & OP & origin & 0 & 1000 & 5/5 & -0.241 & 0.475 & 77 & 58 & 1.44 & 0 \\
 & LL & OP & abs & 0 & 2000 & 5/5 & -0.154 & 0.240 & 67 & 39 & 0.60 & 0 \\
 & LL & OP & origin & 0 & 2000 & 5/5 & -0.288 & 0.409 & 60 & 29 & 0.97 & 0 \\
 & LL & PS-stc & abs & 0 & 500 & 5/5 & 0.025 & 0.872 & 95 & 94 & 3.70 & 0 \\
 & LL & PS-stc & origin & 0 & 500 & 5/5 & -0.136 & 1.071 & 87 & 74 & 3.87 & 0 \\
 & LL & PS-stc & abs & 0 & 1000 & 5/5 & 0.048 & 0.288 & 95 & 92 & 1.02 & 0 \\
 & LL & PS-stc & origin & 0 & 1000 & 5/5 & -0.200 & 0.454 & 79 & 60 & 1.41 & 0 \\
 & LL & PS-stc & abs & 0 & 2000 & 5/5 & 0.041 & 0.197 & 91 & 86 & 0.67 & 0 \\
 & LL & PS-stc & origin & 0 & 2000 & 5/5 & -0.250 & 0.377 & 62 & 30 & 0.93 & 0 \\
 & SL & OP & abs & 0 & 1000 & 5/5 & -0.306 & 0.514 & 84 & 74 & 1.35 & 0 \\
 & SL & OP & origin & 0 & 1000 & 5/5 & 0.309 & 1.010 & 97 & 94 & 2.98 & 0 \\
 & SL & OP & abs & 0 & 2000 & 5/5 & -0.296 & 0.456 & 75 & 58 & 1.11 & 0 \\
 & SL & OP & origin & 0 & 2000 & 5/5 & 0.255 & 0.826 & 96 & 95 & 2.77 & 0 \\
 & SL & PS-stc & abs & 0 & 1000 & 5/5 & -0.113 & 0.488 & 97 & 94 & 1.61 & 0 \\
 & SL & PS-stc & origin & 0 & 1000 & 5/5 & 0.290 & 0.976 & 99 & 98 & 3.00 & 0 \\
 & SL & PS-stc & abs & 0 & 2000 & 5/5 & -0.101 & 0.384 & 94 & 90 & 1.33 & 0 \\
 & SL & PS-stc & origin & 0 & 2000 & 5/5 & 0.245 & 0.825 & 98 & 96 & 2.77 & 0 \\
S2a & LL & OP & abs & 0 & 500 & 3/5 & -0.221 & 0.316 & 95 & 97 & 0.81 & 40 \\
 & LL & OP & origin & 0 & 500 & 5/5 & 0.041 & 0.429 & 95 & 95 & 1.64 & 0 \\
 & LL & OP & abs & 0 & 1000 & 3/5 & -0.232 & 0.284 & 81 & 84 & 0.59 & 40 \\
 & LL & OP & origin & 0 & 1000 & 5/5 & 0.003 & 0.291 & 95 & 95 & 1.11 & 0 \\
 & LL & OP & abs & 0 & 2000 & 3/5 & -0.236 & 0.263 & 56 & 59 & 0.43 & 40 \\
 & LL & OP & origin & 0 & 2000 & 5/5 & -0.004 & 0.202 & 94 & 94 & 0.77 & 0 \\
 & LL & PS-stc & abs & 0 & 500 & 3/5 & -0.004 & 0.297 & 94 & 95 & 1.12 & 40 \\
 & LL & PS-stc & origin & 0 & 500 & 5/5 & 0.078 & 0.509 & 96 & 96 & 1.72 & 0 \\
 & LL & PS-stc & abs & 0 & 1000 & 3/5 & -0.005 & 0.217 & 95 & 96 & 0.78 & 40 \\
 & LL & PS-stc & origin & 0 & 1000 & 5/5 & 0.033 & 0.286 & 96 & 96 & 1.07 & 0 \\
 & LL & PS-stc & abs & 0 & 2000 & 3/5 & -0.003 & 0.152 & 95 & 96 & 0.57 & 40 \\
 & LL & PS-stc & origin & 0 & 2000 & 5/5 & 0.019 & 0.194 & 96 & 96 & 0.73 & 0 \\
S2b & LL & OP & abs & 0 & 500 & 5/5 & -0.249 & 0.395 & 94 & 94 & 1.16 & 0 \\
 & LL & OP & origin & 0 & 500 & 5/5 & 0.011 & 0.381 & 96 & 96 & 1.54 & 0 \\
 & LL & OP & abs & 0 & 1000 & 5/5 & -0.232 & 0.305 & 80 & 78 & 0.77 & 0 \\
 & LL & OP & origin & 0 & 1000 & 5/5 & -0.025 & 0.263 & 95 & 95 & 1.02 & 0 \\
 & LL & OP & abs & 0 & 2000 & 5/5 & -0.220 & 0.260 & 61 & 50 & 0.53 & 0 \\
 & LL & OP & origin & 0 & 2000 & 5/5 & -0.026 & 0.180 & 95 & 95 & 0.70 & 0 \\
 & LL & PS-stc & abs & 0 & 500 & 5/5 & -0.032 & 0.395 & 97 & 97 & 1.41 & 0 \\
 & LL & PS-stc & origin & 0 & 500 & 5/5 & 0.032 & 0.436 & 97 & 98 & 1.59 & 0 \\
 & LL & PS-stc & abs & 0 & 1000 & 5/5 & -0.015 & 0.222 & 96 & 96 & 0.86 & 0 \\
 & LL & PS-stc & origin & 0 & 1000 & 5/5 & 0.013 & 0.248 & 96 & 97 & 0.97 & 0 \\
 & LL & PS-stc & abs & 0 & 2000 & 5/5 & -0.005 & 0.150 & 96 & 95 & 0.59 & 0 \\
 & LL & PS-stc & origin & 0 & 2000 & 5/5 & 0.012 & 0.168 & 96 & 96 & 0.66 & 0 \\
S2c & LL & OP & abs & 0 & 500 & 5/5 & -0.172 & 0.291 & 90 & 87 & 0.82 & 0 \\
 & LL & OP & origin & 0 & 500 & 5/5 & 0.042 & 0.446 & 94 & 94 & 1.60 & 0 \\
 & LL & OP & abs & 0 & 1000 & 5/5 & -0.187 & 0.246 & 77 & 68 & 0.58 & 0 \\
 & LL & OP & origin & 0 & 1000 & 5/5 & 0.016 & 0.291 & 95 & 95 & 1.09 & 0 \\
 & LL & OP & abs & 0 & 2000 & 5/5 & -0.189 & 0.221 & 58 & 38 & 0.41 & 0 \\
 & LL & OP & origin & 0 & 2000 & 5/5 & -0.013 & 0.181 & 95 & 95 & 0.74 & 0 \\
 & LL & PS-stc & abs & 0 & 500 & 5/5 & 0.013 & 0.286 & 96 & 96 & 1.05 & 0 \\
 & LL & PS-stc & origin & 0 & 500 & 5/5 & 0.064 & 0.434 & 96 & 96 & 1.60 & 0 \\
 & LL & PS-stc & abs & 0 & 1000 & 5/5 & 0.003 & 0.198 & 96 & 97 & 0.73 & 0 \\
 & LL & PS-stc & origin & 0 & 1000 & 5/5 & 0.036 & 0.277 & 96 & 96 & 1.05 & 0 \\
 & LL & PS-stc & abs & 0 & 2000 & 5/5 & 0.008 & 0.141 & 95 & 95 & 0.52 & 0 \\
 & LL & PS-stc & origin & 0 & 2000 & 5/5 & 0.007 & 0.174 & 96 & 96 & 0.70 & 0 \\
S3a & LL & OP & abs & 0 & 500 & 15/15 & -0.095 & 2.576 & 94 & 91 & 19.97 & 0 \\
 & LL & OP & origin & 0 & 500 & 15/15 & 0.032 & 6.250 & 92 & 92 & 48.25 & 0 \\
 & LL & OP & abs & 0 & 1000 & 15/15 & -0.104 & 0.241 & 92 & 91 & 0.78 & 0 \\
 & LL & OP & origin & 0 & 1000 & 15/15 & 0.081 & 0.395 & 93 & 93 & 1.37 & 0 \\
 & LL & OP & abs & 0 & 2000 & 15/15 & -0.107 & 0.165 & 84 & 79 & 0.47 & 0 \\
 & LL & OP & origin & 0 & 2000 & 15/15 & 0.008 & 0.216 & 92 & 93 & 0.79 & 0 \\
 & LL & OP & abs & 0.5 & 500 & 15/15 & 0.000 & 0.136 & 97 & 97 & 0.48 & 0 \\
 & LL & OP & origin & 0.5 & 500 & 15/15 & 0.000 & 0.212 & 97 & 96 & 0.80 & 0 \\
 & LL & OP & abs & 0.5 & 1000 & 15/15 & 0.000 & 0.095 & 95 & 93 & 0.32 & 0 \\
 & LL & OP & origin & 0.5 & 1000 & 15/15 & 0.000 & 0.142 & 97 & 96 & 0.54 & 0 \\
 & LL & OP & abs & 0.5 & 2000 & 15/15 & 0.000 & 0.074 & 90 & 83 & 0.22 & 0 \\
 & LL & OP & origin & 0.5 & 2000 & 15/15 & 0.000 & 0.098 & 97 & 96 & 0.37 & 0 \\
 & LL & OP & abs & 1 & 500 & 15/15 & 0.504 & 5.737 & 93 & 88 & 82.17 & 0 \\
 & LL & OP & origin & 1 & 500 & 15/15 & 0.753 & 9.576 & 88 & 88 & 140.68 & 0 \\
 & LL & OP & abs & 1 & 1000 & 15/15 & 0.329 & 1.811 & 93 & 86 & 7.49 & 0 \\
 & LL & OP & origin & 1 & 1000 & 15/15 & 0.318 & 3.074 & 90 & 90 & 12.63 & 0 \\
 & LL & OP & abs & 1 & 2000 & 15/15 & 0.292 & 1.050 & 93 & 83 & 2.67 & 0 \\
 & LL & OP & origin & 1 & 2000 & 15/15 & 0.270 & 1.885 & 95 & 95 & 4.56 & 0 \\
 & LL & PS-stc & abs & 0 & 500 & 15/15 & 0.124 & 6.490 & 96 & 95 & 149.60 & 0 \\
 & LL & PS-stc & origin & 0 & 500 & 15/15 & 0.653 & 14.460 & 94 & 94 & 333.11 & 0 \\
 & LL & PS-stc & abs & 0 & 1000 & 15/15 & 0.004 & 0.760 & 97 & 96 & 3.01 & 0 \\
 & LL & PS-stc & origin & 0 & 1000 & 15/15 & 0.162 & 1.621 & 95 & 95 & 6.16 & 0 \\
 & LL & PS-stc & abs & 0 & 2000 & 15/15 & -0.004 & 0.140 & 95 & 94 & 0.54 & 0 \\
 & LL & PS-stc & origin & 0 & 2000 & 15/15 & 0.020 & 0.204 & 93 & 94 & 0.74 & 0 \\
 & LL & PS-stc & abs & 0.5 & 500 & 15/15 & 0.000 & 0.139 & 98 & 97 & 0.54 & 0 \\
 & LL & PS-stc & origin & 0.5 & 500 & 15/15 & 0.000 & 0.200 & 98 & 97 & 0.79 & 0 \\
 & LL & PS-stc & abs & 0.5 & 1000 & 15/15 & 0.000 & 0.090 & 98 & 97 & 0.34 & 0 \\
 & LL & PS-stc & origin & 0.5 & 1000 & 15/15 & 0.000 & 0.136 & 98 & 97 & 0.50 & 0 \\
 & LL & PS-stc & abs & 0.5 & 2000 & 15/15 & 0.000 & 0.060 & 97 & 96 & 0.23 & 0 \\
 & LL & PS-stc & origin & 0.5 & 2000 & 15/15 & 0.000 & 0.091 & 97 & 95 & 0.34 & 0 \\
 & LL & PS-stc & abs & 1 & 500 & 15/15 & 0.218 & 7.905 & 93 & 90 & 203.48 & 0 \\
 & LL & PS-stc & origin & 1 & 500 & 15/15 & 0.036 & 10.823 & 90 & 90 & 325.66 & 0 \\
 & LL & PS-stc & abs & 1 & 1000 & 15/15 & 0.039 & 1.663 & 93 & 90 & 8.10 & 0 \\
 & LL & PS-stc & origin & 1 & 1000 & 15/15 & 0.095 & 3.071 & 89 & 89 & 14.92 & 0 \\
 & LL & PS-stc & abs & 1 & 2000 & 15/15 & 0.058 & 0.727 & 95 & 94 & 1.89 & 0 \\
 & LL & PS-stc & origin & 1 & 2000 & 15/15 & 0.109 & 1.487 & 94 & 94 & 3.29 & 0 \\
 & SL & OP & abs & 0 & 1000 & 15/15 & -0.102 & 0.373 & 96 & 94 & 1.38 & 0 \\
 & SL & OP & origin & 0 & 1000 & 15/15 & 0.007 & 0.659 & 98 & 98 & 2.25 & 0 \\
 & SL & OP & abs & 0 & 2000 & 15/15 & -0.109 & 0.326 & 92 & 90 & 1.15 & 0 \\
 & SL & OP & origin & 0 & 2000 & 15/15 & 0.108 & 0.648 & 97 & 98 & 2.07 & 0 \\
 & SL & OP & abs & 0.5 & 1000 & 15/15 & 0.001 & 0.191 & 97 & 94 & 0.73 & 0 \\
 & SL & OP & origin & 0.5 & 1000 & 15/15 & 0.057 & 0.359 & 99 & 99 & 1.44 & 0 \\
 & SL & OP & abs & 0.5 & 2000 & 15/15 & -0.002 & 0.140 & 95 & 89 & 0.54 & 0 \\
 & SL & OP & origin & 0.5 & 2000 & 15/15 & 0.033 & 0.290 & 98 & 97 & 1.15 & 0 \\
 & SL & OP & abs & 1 & 1000 & 15/15 & 0.258 & 0.697 & 92 & 60 & 1.82 & 0 \\
 & SL & OP & origin & 1 & 1000 & 15/15 & -0.168 & 0.914 & 94 & 67 & 2.71 & 0 \\
 & SL & OP & abs & 1 & 2000 & 15/15 & 0.257 & 0.644 & 91 & 76 & 1.73 & 0 \\
 & SL & OP & origin & 1 & 2000 & 15/15 & -0.176 & 0.832 & 95 & 86 & 2.59 & 0 \\
 & SL & PS-stc & abs & 0 & 1000 & 15/15 & 0.024 & 0.393 & 97 & 97 & 1.55 & 0 \\
 & SL & PS-stc & origin & 0 & 1000 & 15/15 & 0.001 & 0.643 & 98 & 99 & 2.28 & 0 \\
 & SL & PS-stc & abs & 0 & 2000 & 15/15 & 0.008 & 0.353 & 96 & 96 & 1.31 & 0 \\
 & SL & PS-stc & origin & 0 & 2000 & 15/15 & 0.095 & 0.599 & 98 & 99 & 2.08 & 0 \\
 & SL & PS-stc & abs & 0.5 & 1000 & 15/15 & 0.000 & 0.190 & 98 & 97 & 0.83 & 0 \\
 & SL & PS-stc & origin & 0.5 & 1000 & 15/15 & 0.049 & 0.354 & 99 & 99 & 1.47 & 0 \\
 & SL & PS-stc & abs & 0.5 & 2000 & 15/15 & -0.001 & 0.143 & 97 & 96 & 0.63 & 0 \\
 & SL & PS-stc & origin & 0.5 & 2000 & 15/15 & 0.033 & 0.290 & 99 & 98 & 1.20 & 0 \\
 & SL & PS-stc & abs & 1 & 1000 & 15/15 & 0.112 & 0.666 & 94 & 65 & 1.93 & 0 \\
 & SL & PS-stc & origin & 1 & 1000 & 15/15 & -0.123 & 0.919 & 96 & 66 & 2.76 & 0 \\
 & SL & PS-stc & abs & 1 & 2000 & 15/15 & 0.103 & 0.590 & 95 & 86 & 1.78 & 0 \\
 & SL & PS-stc & origin & 1 & 2000 & 15/15 & -0.146 & 0.831 & 96 & 87 & 2.61 & 0 \\
S3b & LL & OP & abs & 0 & 500 & 15/15 & -0.079 & 1.073 & 93 & 90 & 6.20 & 0 \\
 & LL & OP & origin & 0 & 500 & 15/15 & 0.011 & 2.464 & 81 & 72 & 14.20 & 0 \\
 & LL & OP & abs & 0 & 1000 & 15/15 & -0.081 & 0.216 & 90 & 86 & 0.69 & 0 \\
 & LL & OP & origin & 0 & 1000 & 15/15 & -0.087 & 0.348 & 76 & 64 & 1.15 & 0 \\
 & LL & OP & abs & 0 & 2000 & 15/15 & -0.074 & 0.150 & 83 & 73 & 0.44 & 0 \\
 & LL & OP & origin & 0 & 2000 & 15/15 & -0.117 & 0.242 & 70 & 51 & 0.72 & 0 \\
 & LL & PS-stc & abs & 0 & 500 & 15/15 & 0.026 & 1.471 & 97 & 95 & 8.65 & 0 \\
 & LL & PS-stc & origin & 0 & 500 & 15/15 & 0.125 & 3.154 & 84 & 76 & 18.16 & 0 \\
 & LL & PS-stc & abs & 0 & 1000 & 15/15 & 0.014 & 0.234 & 96 & 94 & 0.84 & 0 \\
 & LL & PS-stc & origin & 0 & 1000 & 15/15 & -0.069 & 0.338 & 78 & 66 & 1.14 & 0 \\
 & LL & PS-stc & abs & 0 & 2000 & 15/15 & 0.024 & 0.146 & 95 & 93 & 0.53 & 0 \\
 & LL & PS-stc & origin & 0 & 2000 & 15/15 & -0.098 & 0.230 & 72 & 52 & 0.69 & 0 \\
 & SL & OP & abs & 0 & 1000 & 15/15 & -0.106 & 0.400 & 95 & 93 & 1.41 & 0 \\
 & SL & OP & origin & 0 & 1000 & 15/15 & 0.001 & 0.729 & 98 & 99 & 2.28 & 0 \\
 & SL & OP & abs & 0 & 2000 & 15/15 & -0.128 & 0.307 & 93 & 87 & 1.11 & 0 \\
 & SL & OP & origin & 0 & 2000 & 15/15 & -0.008 & 0.578 & 97 & 98 & 2.04 & 0 \\
 & SL & PS-stc & abs & 0 & 1000 & 15/15 & 0.017 & 0.407 & 98 & 98 & 1.56 & 0 \\
 & SL & PS-stc & origin & 0 & 1000 & 15/15 & -0.028 & 0.674 & 99 & 100 & 2.33 & 0 \\
 & SL & PS-stc & abs & 0 & 2000 & 15/15 & -0.023 & 0.305 & 97 & 96 & 1.23 & 0 \\
 & SL & PS-stc & origin & 0 & 2000 & 15/15 & -0.024 & 0.524 & 98 & 99 & 2.05 & 0 \\
\hline
\end{longtable}
\endgroup
\begin{longtable}{lllrrrrrrrr}
\caption{Scalar slope/intercept calibration for the linear learner (LL) in the linear scenarios (S1a, S2).
$\phi$ and $\alpha$ are the slope and level-shift scalars; rej.\ is the rejection rate of $H_0:\phi=1$ (power when $\phi_{\rm true}\ne1$).
Cov is Wald 95\% coverage (\%).}
\label{tab:sim-comp-params}\\
\hline
Scn & Axis & Policy & $n$ & $\phi$ bias & $\phi$ RMSE & $\phi$ Cov & rej.\ & $\alpha$ bias & $\alpha$ RMSE & $\alpha$ Cov \\
\hline\endfirsthead
\hline\endhead
S1a & obs & OP & 500 & 0.039 & 0.274 & 95 & 36 & -0.065 & 0.478 & 96 \\
 & obs & OP & 1000 & 0.019 & 0.171 & 96 & 54 & -0.029 & 0.303 & 95 \\
 & obs & OP & 2000 & 0.001 & 0.113 & 96 & 79 & 0.001 & 0.202 & 95 \\
 & treat & OP & 500 & 0.025 & 0.272 & 96 & 31 & -0.248 & 0.556 & 98 \\
 & treat & OP & 1000 & 0.002 & 0.187 & 96 & 51 & -0.209 & 0.406 & 97 \\
 & treat & OP & 2000 & -0.003 & 0.126 & 95 & 75 & -0.199 & 0.309 & 92 \\
 & treat & PS-stc & 500 & 0.041 & 0.274 & 96 & 33 & -0.073 & 0.532 & 96 \\
 & treat & PS-stc & 1000 & 0.017 & 0.182 & 97 & 53 & -0.030 & 0.350 & 96 \\
 & treat & PS-stc & 2000 & 0.011 & 0.119 & 96 & 74 & -0.018 & 0.230 & 96 \\
S2a & obs & OP & 500 & 0.043 & 0.270 & 95 & 32 & -0.090 & 0.581 & 95 \\
 & obs & OP & 1000 & 0.014 & 0.168 & 96 & 51 & -0.028 & 0.357 & 96 \\
 & obs & OP & 2000 & 0.008 & 0.121 & 95 & 77 & -0.013 & 0.260 & 95 \\
 & treat & OP & 500 & 0.027 & 0.286 & 95 & 31 & -0.322 & 0.727 & 98 \\
 & treat & OP & 1000 & 0.002 & 0.194 & 95 & 47 & -0.261 & 0.510 & 97 \\
 & treat & OP & 2000 & -0.002 & 0.135 & 94 & 70 & -0.245 & 0.392 & 93 \\
 & treat & PS-stc & 500 & 0.052 & 0.340 & 96 & 31 & -0.146 & 0.861 & 95 \\
 & treat & PS-stc & 1000 & 0.022 & 0.190 & 96 & 50 & -0.059 & 0.461 & 96 \\
 & treat & PS-stc & 2000 & 0.013 & 0.129 & 96 & 72 & -0.031 & 0.307 & 95 \\
S2b & obs & OP & 500 & 0.027 & 0.244 & 96 & 64 & -0.049 & 0.478 & 96 \\
 & obs & OP & 1000 & 0.014 & 0.161 & 96 & 85 & -0.031 & 0.312 & 94 \\
 & obs & OP & 2000 & 0.004 & 0.108 & 95 & 99 & -0.008 & 0.205 & 95 \\
 & treat & OP & 500 & 0.007 & 0.254 & 96 & 58 & -0.260 & 0.591 & 98 \\
 & treat & OP & 1000 & -0.017 & 0.176 & 95 & 85 & -0.207 & 0.403 & 96 \\
 & treat & OP & 2000 & -0.018 & 0.120 & 95 & 98 & -0.194 & 0.311 & 93 \\
 & treat & PS-stc & 500 & 0.021 & 0.291 & 97 & 59 & -0.064 & 0.651 & 97 \\
 & treat & PS-stc & 1000 & 0.009 & 0.166 & 96 & 84 & -0.028 & 0.340 & 97 \\
 & treat & PS-stc & 2000 & 0.008 & 0.112 & 96 & 98 & -0.017 & 0.241 & 96 \\
S2c & obs & OP & 500 & 0.029 & 0.256 & 94 & 36 & -0.050 & 0.527 & 95 \\
 & obs & OP & 1000 & 0.020 & 0.171 & 96 & 52 & -0.032 & 0.352 & 96 \\
 & obs & OP & 2000 & 0.006 & 0.116 & 96 & 79 & -0.012 & 0.237 & 96 \\
 & treat & OP & 500 & 0.028 & 0.298 & 94 & 33 & -0.280 & 0.705 & 97 \\
 & treat & OP & 1000 & 0.011 & 0.194 & 95 & 50 & -0.246 & 0.481 & 96 \\
 & treat & OP & 2000 & -0.009 & 0.121 & 95 & 76 & -0.201 & 0.333 & 94 \\
 & treat & PS-stc & 500 & 0.043 & 0.291 & 96 & 34 & -0.092 & 0.652 & 95 \\
 & treat & PS-stc & 1000 & 0.024 & 0.184 & 96 & 50 & -0.052 & 0.402 & 96 \\
 & treat & PS-stc & 2000 & 0.005 & 0.116 & 96 & 76 & -0.008 & 0.268 & 95 \\
\hline
\end{longtable}

\end{document}